\documentclass[a4paper, UKenglish, cleveref, autoref]{lipics-v2021}

\title{Expressive Power of Property Graph Constraint Languages} 

\author{Stefania Dumbrava}
{
Univ.~Paris Cité, Inria, CNRS, IRIF, F-75013, Paris, France \and
ENSIIE, SAMOVAR, Télécom SudParis, Institut Polytechnique de Paris, France \and \url{https://web4.ensiie.fr/~stefania.dumbrava/}}
{stefania.dumbrava@ensiie.fr}
{https://orcid.org/0000-0002-6664-0620}
{}

\author{Nadime Francis}
{Univ.~Gustave Eiffel, CNRS, LIGM, F-77454 Marne-la-Vallée, France \and \url{https://pagespro.univ-gustave-eiffel.fr/nadime-francis}}
{nadime.francis@univ-eiffel.fr}
{https://orcid.org/0009-0009-4531-7435}
{}

\author{Victor Marsault}
{Univ.~Gustave Eiffel, CNRS, LIGM, F-77454 Marne-la-Vallée, France
\and \url{https://victor.marsault.xyz}}
{victor.marsault@univ-eiffel.fr}
{https://orcid.org/0000-0002-2325-6004}
{}

\author{Steven Sailly}
{SAMOVAR, Télécom SudParis, Institut Polytechnique de Paris, France \and \url{https://steven.saill.yt/}}
{steven.sailly@ip-paris.fr}
{https://orcid.org/0009-0005-8215-7996}
{}

\authorrunning{S.~Dumbrava, N.~Francis, V.~Marsault and S.~Sailly}

\Copyright{Stefania Dumbrava, Nadime Francis, Victor Marsault and Steven Sailly} 



\category{} 

\funding{This work was partially supported by the grant ANR-24-CE25-1109.}

\EventEditors{}
\EventNoEds{1}
\EventLongTitle{}
\EventShortTitle{}
\EventAcronym{}
\EventYear{}
\EventDate{}
\EventLocation{}
\EventLogo{}
\SeriesVolume{}
\ArticleNo{}

\usepackage{tikz}
\usetikzlibrary{
    arrows.meta,
    intersections,
    positioning,
    shapes.multipart,
    calc, 
    shapes.misc
}
\tikzset{>={Stealth[scale=1]}}
\tikzset{footnote/.style={
    circle, inner sep=0pt, minimum width=13pt, draw=black, fill=white, font=\footnotesize}}

\colorlet{ourpurple}{gray!30!blue!20}
\colorlet{ourbrown}{orange!20!gray!35}
\colorlet{ourteal}{teal!60!gray!30}
\colorlet{ouryellow}{yellow!50!brown!30}
\colorlet{ourpink}{purple!60!gray!20}

\usepackage{colortbl}
\usepackage{braket}
\usepackage{diagbox}
\usepackage{mathrsfs}
\usepackage{adjustbox}
\usepackage[linesnumbered,ruled,vlined]{algorithm2e}
\SetAlgoNoEnd

\usepackage[bibliography=common]{apxproof}

\newtheoremrep{proposition}[theorem]{Proposition}
\newtheoremrep{lemma}[theorem]{Lemma}

\usepackage{listings}
\lstdefinelanguage{pgkeysLang}
{
  morekeywords={
    FOR,
    WITHIN,
    EXCLUSIVE,
    MANDATORY,
    SINGLETON,
    IDENTIFIER,
    MATCH,
    WHERE,
    NOT,
    OR,
    AND,
    EXISTS,
    RETURN,
    MATCH,
    IN,
    IS,
    NULL,
    hasKey,
    FunctionalProperty,
    context,
    test,
    unique,
    key,
    selector,
    field,
    allInstances,
    isUnique,
    Tuple,
    class,
    extent,
    relationship,
    inverse,
    attribute
  },
  sensitive=false, 
  morecomment=[l]{//}, 
  morecomment=[s]{/*}{*/}, 
  morestring=[b]" 
}

\usepackage{soul}
\usepackage{color}
\colorlet{hlgray}{black!10}
\let\highlight\relax 
\newcommand{\highlight}[2][hlgray]{\mathchoice%
  {\colorbox{#1}{$\rule[-.2em]{0pt}{.9em}\displaystyle#2$}}%
  {\colorbox{#1}{$\rule[-.2em]{0pt}{.9em}\textstyle#2$}}%
  {\colorbox{#1}{$\rule[-.2em]{0pt}{.9em}\scriptstyle#2$}}%
  {\colorbox{#1}{$\rule[-.2em]{0pt}{.9em}\scriptscriptstyle#2$}}}%
\definecolor{eclipseBlue}{RGB}{42,0.0,255}
\definecolor{eclipseGreen}{RGB}{63,127,95}
\definecolor{eclipsePurple}{RGB}{127,0,85}
 
\newcommand{\fragmentfont}[1]{\ensuremath{\text{{\rm\sc{}#1}}}}
\newcommand{\pgkey}{\fragmentfont{PG-Key}}
\newcommand{\pgkeys}{\fragmentfont{PG-Keys}}
\newcommand{\pgschema}{\fragmentfont{PG-Schema}}
\newcommand{\mpgkey}{\fragmentfont{mPG-Key}}
\newcommand{\mpgkeys}{\fragmentfont{mPG-Keys}}
\newcommand{\ged}{\fragmentfont{GED}}
\newcommand{\gfd}{\fragmentfont{GFD}}
\newcommand{\gfds}{\fragmentfont{GFDs}}
\newcommand{\gdd}{\fragmentfont{GDD}}
\newcommand{\ggd}{\fragmentfont{GGD}}
\newcommand{\ggds}{\fragmentfont{GGDs}}
\newcommand{\tgd}{\fragmentfont{TGD}}
\newcommand{\egd}{\fragmentfont{EGD}}

\renewcommand\set[1]{\csname set \endcsname{\mskip-\medmuskip#1\mskip-\medmuskip}}
\renewcommand{\bar}[1]{\overline{#1}}

\newcommand{\splbl}{\boldsymbol{\ell_0}}
\newcommand{\Ell}{\mathscr{L}}

\newcommand{\doteqc}{\mathrel{=_c}}
\newcommand{\ndoteqc}{\mathrel{\not\doteqc}}

\newcommand{\cq}{\text{\rm\sc CQ}}
\newcommand{\cqs}{\text{\rm\sc CQs}}
\newcommand{\cqeq}{\ensuremath{\cq[=]}}
\newcommand{\cqneq}{\ensuremath{\cq[=,\neq]}}

\newcommand{\rpq}{\text{\rm\sc RPQ}}

\newcommand{\crpq}{\text{\rm\sc CRPQ}}
\newcommand{\crpqs}{\text{\rm\sc CRPQs}}
\newcommand{\crpqeq}{\ensuremath{{\crpq[=]}}}
\newcommand{\crpqeqc}{\ensuremath{{\crpq[\doteqc]}}}
\newcommand{\crpqneq}{\ensuremath{{\crpq[=, \neq]}}}
\newcommand{\crpqneqc}{\ensuremath{{\crpq[=_c, \neq_c]}}}
\newcommand{\lstinmath}[1]{\text{\lstinline{#1}}}

\newcommand{\conp}{\textbf{coNP}}
\newcommand{\piip}{\ensuremath{\mathbf{\Pi_2^P}}}
\newcommand{\diip}{\ensuremath{\mathbf{\Delta_2^P}}}
\newcommand{\tiip}{\ensuremath{\mathbf{\Theta_2^P}}}

\DeclareMathOperator{\im}{\mathsf{Im}}
\DeclareMathOperator{\dom}{\mathsf{Dom}}

\DeclareMathOperator{\obj}{\mathsf{Obj}}
\DeclareMathOperator{\lab}{\mathsf{Lab}}
\DeclareMathOperator{\Lab}{\lab}
\DeclareMathOperator{\key}{\mathsf{Key}}

\DeclareMathOperator{\val}{\mathsf{Val}}

\DeclareMathOperator{\var}{\mathsf{Var}}
\DeclareMathOperator{\Var}{\var}
\DeclareMathOperator{\reg}{\mathsf{RegExp}}
\DeclareMathOperator{\paths}{\mathsf{Paths}}
\DeclareMathOperator{\ex}{\mathsf{exists}}

\newcommand{\pginmathkw}[1]{\mathrel{\lstinmath{#1}}}
\newcommand{\FOR}{\pginmathkw{FOR}}

\newcommand{\WITHIN}{\pginmathkw{WITHIN}}

\newcommand{\EXCLUSIVE}{\pginmathkw{EXCLUSIVE}}
\newcommand{\MANDATORY}{\pginmathkw{MANDATORY}}
\newcommand{\SINGLETON}{\pginmathkw{SINGLETON}}

\newcommand{\WHERE}{\pginmathkw{WHERE}}
\newcommand{\ALL}{{\color{eclipsePurple}\pginmathkw{ALL}}}
\newcommand{\WALK}{{\color{eclipsePurple}\pginmathkw{WALK}}}

\newcommand{\enumstyle}[1]{\textcolor{lipicsGray}{\sffamily\bfseries\upshape\mathversion{bold}#1}}

\newtheorem{property}[theorem]{Property}

\hideLIPIcs 
\nolinenumbers
\begin{document}

\maketitle

\begin{abstract}
We present the first principled and systematic study of the expressive power of property graph constraint languages, focused on the recent PG-Keys language, set to inform the upcoming revision of the GQL standard. To this end, we position PG-Keys within the broader landscape of existing formalisms.
In particular, we compare PG-Keys with two core property graph constraint languages: Graph Functional Dependencies (GFD) and Graph Generating Dependencies (GGD). 
One hurdle is that these formalisms allow different kinds of graph pattern languages and predicates. To make a fair comparison, based on their structural differences only, we first present a unifying framework.
Within this framework, we consider conjunctive regular path queries (CRPQ) as graph patterns with equality and inequality predicates. 
We then identify well-behaved fragments, establish expressiveness inclusion, and prove separation results, yielding a complete and strict hierarchy of expressive power.
The results identify precisely when PG-Keys provide strictly greater expressive power, clarifying their place among state-of-the-art property graph constraint formalisms.
\end{abstract}

\section{Introduction}\label{sec:intro}
Integrity constraints are central to ensuring data quality and supporting reliable
data integration and exchange. Constraint languages have been
extensively investigated~\cite{DBLP:books/aw/AbiteboulHV95,DBLP:journals/jacm/Fagin82} for relational databases,
leading to a deep understanding of their expressive power, complexity, and impact on
query processing. Conversely, the comparative expressiveness of constraint languages for \emph{graph databases}
remains underexplored, as the interplay between graph topology and data values introduces new modeling and reasoning challenges.

The most generic form of integrity constraints is given by general dependencies~\cite{DBLP:conf/icalp/BeeriV81}, defined as being of the form~$\forall \bar{x} \forall \bar{y} \big(\varphi(\bar{x}, \bar{y}) \Rightarrow \exists \bar{z} \psi(\bar{y}, \bar{z})\big)$. The body of the formula, $\varphi$, is a non-empty conjunction of relational atoms and its head, $\psi$, is either a single equality atom, for equality generating dependencies ($\egd$), or a non-empty conjunction of relational atoms, for tuple generating dependencies ($\tgd$). When considering graph dependencies, relations are restricted to arity two. Building on this foundation, \emph{Graph Functional Dependencies} ($\gfd$)~\cite{DBLP:conf/sigmod/FanWX16a} and \emph{Graph Generating Dependencies} ($\ggd$)~\cite{DBLP:conf/cikm/ShimomuraFY20} adapt $\egd$ and $\tgd{}$, respectively.

Property graphs~\cite{DBLP:conf/amw/Angles18} are widely used as data models for graph databases, representing data as multi-labeled multigraphs with key-value properties on both nodes and edges. As their adoption for highly interconnected and semantically rich data grows, several dedicated constraint formalisms have been proposed.

Recently, the \pgkeys{} language \cite{DBLP:conf/sigmod/AnglesBDFHHLLLM21} was introduced to constrain property graphs and to reference and identify their objects. A \pgkey{} associates a scope with a descriptor and a combination of three assertion keywords: $\MANDATORY$, $\EXCLUSIVE$ and/or $\SINGLETON$. The scope defines the objects to which the constraint applies and the descriptor how to obtain a key value from each such object. $\MANDATORY$ requires that each object in scope has at least one key value, $\EXCLUSIVE$  that no two objects in scope share a key value, and $\SINGLETON$ that each object in scope has at most one key value. Despite growing interest, the landscape of property graph constraint languages remains fragmented. Manouvrier and Belhajjame~\cite{DBLP:conf/adbis/ManouvrierB24,DBLP:journals/is/ManouvrierB26} implement translations of constraint languages, including \ged{} and \ggd, into an overapproximation of \pgkeys{} (see \cref{sec:related}). To our knowledge, \emph{no prior work has formally positioned \pgkeys{} relative to these formalisms}.

Understanding how their core features relate is crucial to establishing expressiveness boundaries and further informing the recent standardized property graph languages GQL~\cite{GQL-ISO} and the 
SQL/PGQ extension~\cite{SQL-PGQ} of the SQL standard.
This is especially relevant as \pgschema~\cite{DBLP:journals/pacmmod/AnglesBD0GHLLMM23}, a superset of \pgkeys, is a candidate to be integrated in the upcoming second edition of GQL.

\begingroup
\newcounter{figurefootnote}
\newcommand{\nextfigfootnote}{\stepcounter{figurefootnote}\thefigurefootnote}
    
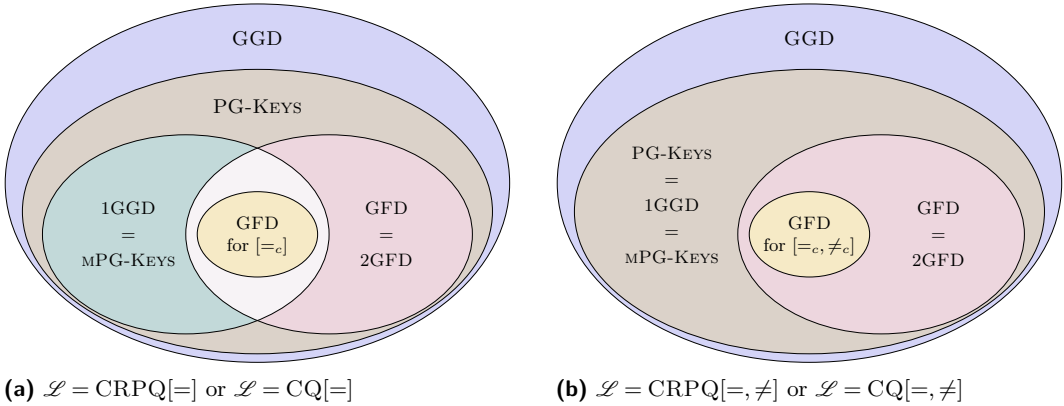
\begin{figure}
    \newlength{\vmtmp}\setlength{\vmtmp}{\dimexpr.5\linewidth-1em\relax}
    \subfloat[$\Ell\mathbin=\crpqeq$ or $\Ell\mathbin=\cqeq$]{%
        \label{fig:venn-eq}%
        \resizebox{\vmtmp}{!}{\begin{tikzpicture}
    \begin{scope}[draw=black]
        \draw[fill=ourpurple, name path=ggd] (0,0.4) ellipse (4.4 and 3.15);
        \draw[fill=ourbrown, name path=pgk] (0,-0.1) ellipse (4.1 and 2.5);
        \begin{scope}[blend group=overlay]
            \fill[fill=ourteal, name path=gfd] (-1.25,-0.5) coordinate (tealcenter) ellipse (2.5 and 1.75);
            \fill[fill=ourpink, name path=1ggd] (1.25,-0.5) coordinate (purplecenter) ellipse (2.5 and 1.75);
        \end{scope}
         \draw (-1.25,-0.5) coordinate ellipse (2.5 and 1.75);
         \draw (1.25,-0.5) coordinate ellipse (2.5 and 1.75);
         
        \draw[fill=ouryellow] ($.5*(tealcenter)+.5*(purplecenter)$) coordinate (yellowcenter) ellipse (1.05 and 0.75);

        \node at (0,2.95) (ggd_txt) {$\ggd$};
        \node at (0,1.75) (pgk_txt) {\pgkeys};
        \path (purplecenter) ++(1,0) node (gfd_txt) {\begin{tabular}{c}\gfd\\$=$\\$2\gfd$\end{tabular}};
        \node at (yellowcenter) (gfdc_txt) {\begin{tabular}{c}$\gfd$\\[-2pt]for $[\doteqc]$\end{tabular}};
        \path (tealcenter) ++(-1,0) node (1ggd_txt) 
            {\begin{tabular}{c}
                $1\ggd$\\$=$\\\mpgkeys\end{tabular}};

        \path (0,2.4) node[footnote] {\nextfigfootnote};
        \path (-1.25,1.25) node[footnote] {\nextfigfootnote};
        \path (1.25,1.25) node[footnote] {\nextfigfootnote};
        \path (0,.8) node[footnote] {\nextfigfootnote};
        \path (-1.75,-0.45) node[footnote] {\nextfigfootnote};
        \path (2.7,-0.45) node[footnote] {\nextfigfootnote};
        \path (0,-1.25) node[footnote] {\nextfigfootnote};
                
    \end{scope}
\end{tikzpicture}}%
    }%
    \hfill%
    \subfloat[$\Ell\mathbin=\crpqneq$ or $\Ell\mathbin=\cqneq$]{%
        \label{fig:venn-neq}%
        \resizebox{\vmtmp}{!}{\begin{tikzpicture}
    \begin{scope}[draw=black]
        \draw[fill=ourpurple, name path=ggd] (0,0.4) ellipse (4.4 and 3.15);
        \draw[fill=ourbrown, name path=pgk] (0,-0.1) ellipse (4.1 and 2.5);
        \draw[fill=ourpink, name path=1ggd] (1.25,-0.5) coordinate (purplecenter) ellipse (2.5 and 1.75);
         
        \draw[fill=ouryellow] ($.5*(tealcenter)+.5*(purplecenter)$) coordinate (yellowcenter) ellipse (1.05 and 0.75);

        \node at (0,2.95) (ggd_txt) {$\ggd$};
        \node at (-2.4,0) (pgk_txt) {\begin{tabular}{c}$\pgkeys$\\$=$\\ $1\ggd$\\ $=$ \\$\mpgkeys$\end{tabular}};
        \path (purplecenter) ++ (1,0) node (gfd_txt) {\begin{tabular}{c}\gfd\\$=$\\$2\gfd$\end{tabular}};
        \node at (yellowcenter) (gfdc_txt) {\begin{tabular}{c}$\gfd$\\[-2pt]for $[\doteqc, \ndoteqc]$\end{tabular}};

        \path (0,2.4) node[footnote] {1};
        \path (-1.9,0.5) node[footnote] {\nextfigfootnote};
        \path (-1.9,-0.4) node[footnote] {5};
        \path (1.25,1.25) node[footnote] {\nextfigfootnote};
        \path (2.7,-0.45) node[footnote] {6};
        \path (0,-1.25) node[footnote] {7};
        
    \end{scope}
\end{tikzpicture}}%
    }%

    \medskip

    \noindent%
    \setcounter{figurefootnote}{0}%
\begin{tabularx}{\linewidth}{>{\footnotesize\raggedleft}r@{\hskip5pt}>{\footnotesize}l>{\footnotesize}X}
    \hline
    \nextfigfootnote. & $\pgkeys{} \subsetneq \ggd$:&Propositions\,\ref{p:pgk-sub-ggd-cq-eq} (inclusion) and~\ref{p:ggd-neq-pgk-cq-eq} (separation) \\
    \rowcolor{black!10}\nextfigfootnote. & $\mpgkeys{} \subsetneq \pgkeys$:&\cref{p:mpgk-neq-pgk-cq-eq} (separation) and  inclusion is immediate \\
    \nextfigfootnote. & $\gfd{} \subsetneq \pgkeys$:&Propositions~\ref{p:gfd-sub-pgk-cq-eq} (inclusion) and  \ref{p:closure-induced-subgraph}.\ref{p:closure-induced-subgraph:sep} (separation)\\
    \rowcolor{black!10}\nextfigfootnote. & $1\ggd{}$ and $\gfd$ are incomparable:&Propositions~\ref{p:closure-induced-subgraph}.\ref{p:closure-induced-subgraph:sep} ($1\ggd\not\subseteq \gfd$) and~\ref{p:ggd-neq-pgk-cq-eq} ($1\ggd\not\supseteq \gfd$)\\
    \nextfigfootnote. & $1\ggd{} = \mpgkeys$:&Simple syntax rewriting (\cref{p:mpgk-is-1ggd-cq-eq}) \\
    \rowcolor{black!10}\nextfigfootnote. & $\gfd = 2\gfd $:&\cref{l:split-gfd}\\
    \nextfigfootnote. & $\gfd$ for $[\doteqc] \subsetneq (1\ggd \cap \gfd)$:&Propositions~\ref{p:gfdc-sub-ggd-cq-eq} (inclusion) and \ref{p:gfdc-neq-gfd-1ggd} (separation)\\
    \rowcolor{black!10}\nextfigfootnote. & $\pgkeys = 1\ggd$:&\cref{p:mpgkey=pgkey}\\
    \nextfigfootnote. & $\gfd \subsetneq 1\ggd$:&Propositions~\ref{p:gfd-sub-1ggq-nq-neq} (inclusion) and \ref{p:closure-induced-subgraph}.\ref{p:closure-induced-subgraph:sep} (separation)\\
    \hline
\end{tabularx}
    \hspace*{0cm plus 1fill}%
    
    \caption{Strict dependency inclusions for constraint language for query language~$\Ell$.}
    \label{fig:inclusion-cqeq-cqneq}
\end{figure}

\subparagraph*{{Contributions.}} We address the existing gap in the literature by providing the first formal comparison of these languages, laying the foundations for their principled design and practical usage. We make the following contributions, which apply regardless of whether constraint patterns are expressed as Conjunctive Queries (\cq{}) or Conjunctive Regular Path Queries (\crpq{}).
\begin{enumerate}
\item \emph{Fine-grained language feature analysis}. We analyze how the equality and inequality of vertices and edges, along with the number of variables that are \emph{shared} between the source\footnotemark\addtocounter{footnote}{-1} and the target\footnote{These are sometimes called the \emph{scope} and \emph{descriptor} of the key, or as the \emph{body} and \emph{head} of the constraint.}, impact the comparative expressive power of the constraint languages. 
This reveals well-behaved fragments, namely $1\ggd$ and $2\gfd$, and clarifies the effect of one important design choice of \pgkeys{}: only one variable can be shared between scope and descriptor.

\item \emph{Expressiveness inclusions}. We define the conditions under which a set of constraints in one language can be translated into a set of constraints in another and establish expressiveness inclusions
(see Figure~\ref{fig:inclusion-cqeq-cqneq}).

\item \emph{Separation results}. We prove that certain constraints are inherently inexpressible in some languages or fragments,  even under extensions, thus precisely characterizing their expressiveness limitations. 
Our results show that all inclusions in Figure~\ref{fig:inclusion-cqeq-cqneq} are strict.
\end{enumerate}

We highlight the most unexpected results. Because \ggd{} allows sharing multiple variables between the body and the head of a constraint, it is possible to simulate the $\SINGLETON$ and $\EXCLUSIVE$ keywords of $\pgkeys$, which results in $\pgkeys{} \subseteq \ggd{}$ (\cref{fig:inclusion-cqeq-cqneq}, Item~1). Moreover, when inequality between edges, vertices, and properties is allowed, a single shared variable is enough to perform this translation. This motivates the definition of $1\ggd$, as a strict syntactic restriction of $\ggd$ that precisely captures $\pgkeys$: $\pgkeys{} = 1\ggd{}$ in that case (\cref{fig:inclusion-cqeq-cqneq}, Item~8). Finally, we show how the $\SINGLETON$ keyword can simulate multiple shared variables in some restricted cases, leading to $\gfd{} \subseteq \pgkeys{}$ (\cref{fig:inclusion-cqeq-cqneq}, Item~3).

\subparagraph*{Outline.} The paper is organized as follows. \cref{sec:preliminaries} introduces the property graph model, the underlying query languages (\cq{} and \crpq{} with equality and inequality), and recasts \gfd, \ggd, and \pgkeys{} in a unified parametric framework. \cref{sec:gen} identifies core fragments, in particular those based on the number~$n$ of shared variables ($n$\gfd{} and $n$\ggd{}) and on the allowed keywords (\mpgkeys{} allows only the keyword $\MANDATORY$). We also establish in \cref{sec:gen} that the $n$\ggd{} hierarchy is strict while the $n$\gfd{} hierarchy collapses. \cref{sec:cq:eq,sec:cq:neq} provide a complete expressiveness analysis under two settings: queries with equality only, and queries with equality and inequality. For each, we establish strict inclusion hierarchies via translations and separation results, leading up to the collapse of \pgkeys{} with 1\ggd{} with equality and inequality. \cref{sec:related} discusses related work, and \cref{sec:conclusion} concludes with implications for constraint design, complexity, and future work. Most proofs are deferred to the appendix due to space constraints.

\section{Preliminaries}\label{sec:preliminaries}
We introduce the notations and definitions used henceforth. For any positive integer $n$, let $[n]$ denote $\set{1, 2, \ldots, n}$. 
We fix countably infinite and pairwise distinct sets $\obj$, $\lab$, $\key$, $\val$, and $\var$ of objects, labels, keys, values, and variables, respectively.
$\reg(\lab)$ denotes the set of all regular expressions built from $\lab$
and, given a regular expression~$r\in\reg(\lab)$, $L(r)$ denotes the language described by~$r$ with the usual definition.

\subsection{Data Model}
Databases in this document are modeled as \emph{property graphs}; that is, multi-edge, multi-label graphs whose vertices and edges additionally carry sets of key-value pairs called \emph{properties}.

\begin{definition}\label{def:graph}
    A \emph{(property) graph} is a tuple $G = (V, E, \eta, \lambda, \pi)$ where $V \subset \obj$ is a finite set of vertices, $E \subset \obj$ is a finite set of edges with $V \cap E = \emptyset$, $\eta: E \to V \times V$ is the incidence function assigning source and target vertices to each edge, $\lambda: V \cup E \to 2^{\lab}$ is a labeling function such that, for each $o\in V\cup E$, $\lambda(o)$ is finite, and $\pi: (V \cup E)\times \key \rightharpoonup \val$ is a partial function that assigns a value to vertex or edge properties, such that it has a finite domain. When the graph is clear from context, we may use $u.a$ as a shorthand for $\pi(u,a)$.
\end{definition}

A graph that will be used as a running example throughout the document is given in Figure~\ref{fig:running-example-graph}.    
We let~$\paths(G)$ denote the set of \emph{paths} in~$G$, that is, nonempty sequences of edges~$e_1\cdots e_n$ such that 
for every~$i\in [n-1]$, the target of $e_i$ coincides with the source of $e_{i+1}$.
We extend~$\eta$ and~$\lambda$ to~$\paths(G)$:
$\eta(e_1\cdots e_n) = (s,t)$ where $s$ is the source of $e_1$ and $t$ is the target of $e_n$; and~$\lambda(e_1\cdots e_n)=\{a_1\cdots a_n\mid a_1\in\lambda(e_1),\ldots,a_n\in\lambda(e_n)\}$.  

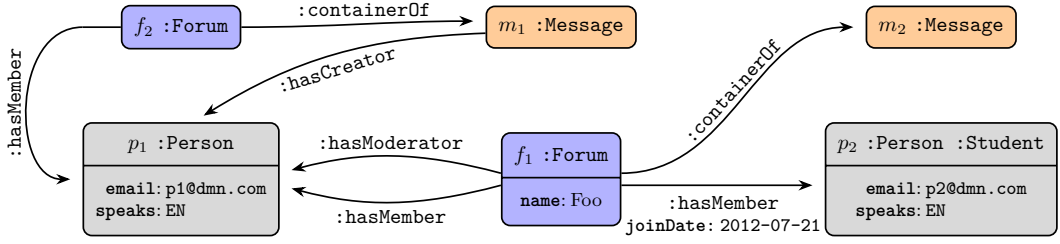
\begin{figure}
    \adjustbox{width=\linewidth}{\begin{tikzpicture}[
    Vertex/.style = {
        rectangle split,
        rectangle split parts = 2,
        draw = black,
        rounded corners,
        thick,
        y=9mm,
        x=12.5mm,
        inner sep=2mm,
    },
    Forum/.style={fill=blue!30},
    Message/.style={fill=orange!40},
    Person/.style={fill=gray!30},
    every edge/.style = {
        ->,
        draw = black,
        thick,
        shorten >= 2mm
    },
    bend angle=15,
    ]\renewcommand{\arraystretch}{.8}
    
    \newcommand{\vmtmpdist}{.1}
    
    \node[Vertex,Forum] (f1) at (0, 0) {
        $f_1$ \texttt{:Forum}
        \nodepart[inner sep=0pt]{second}
        \begin{tabular}{@{}r@{}l@{}}
            \textbf{\texttt{name}}:\hspace*{.5ex}& Foo 
        \end{tabular}
    };
    
    \node[Vertex,Person] (p1) at (-5, 0) {
        $p_1$ \texttt{:Person}
        \nodepart{second}
        \begin{tabular}{@{}r@{}l@{}}
            \textbf{\texttt{email}}:\hspace*{.5ex}& \texttt{p1@dmn.com} \\
            \textbf{\texttt{speaks}}:\hspace*{.5ex}& \texttt{EN}
        \end{tabular}%
    };
    
    \node[Vertex,Person] (p2) at (6, 0) {
        $p_2$ \texttt{:Person :Student}
        \nodepart{second}
        \begin{tabular}{@{}r@{}l@{}}
            \textbf{\texttt{email}}:\hspace*{.5ex}& \texttt{p2@dmn.com} \\
            \textbf{\texttt{speaks}}:\hspace*{.5ex}& \texttt{EN}
        \end{tabular}
    };
    
    \node[Vertex,Forum,rectangle] (f2) at (-5, 2.8) {
        $f_2$ \texttt{:Forum}
    };
    
    \node[Vertex,Message,rectangle] (m1) at (0, 2.8) {
        $m_1$ \texttt{:Message}
    };
    
    \node[Vertex,Message,rectangle] (m2) at (6, 2.8) {
        $m_2$ \texttt{:Message}
    };

    \draw ($(f2.east)$) edge[out=0,in=180] node[above] {\texttt{:containerOf}} ($(m1.west)+(0,\vmtmpdist)$);

    \draw[thick] (f2.west) -- (p1.west|-f2.west) edge[out=180,in=180,looseness=1.2] node[sloped, above,rotate=180] {\texttt{:hasMember}} (p1);
    
    \draw ($(m1.west)+(0,-\vmtmpdist)$) edge[bend right] node[sloped, below] {\texttt{:hasCreator}} ($(p1.north)+(\vmtmpdist*2,0)$);
    
    \draw ($(f1.west)+(0,\vmtmpdist)$) edge[bend right] node[above] {\texttt{:hasModerator}} ($(p1.east)+(0,\vmtmpdist)$);

    \draw ($(f1.west)+(0,-\vmtmpdist)$) edge[bend left] node[below] {\texttt{:hasMember}} ($(p1.east)+(0,-\vmtmpdist)$);
    
    \draw ($(f1.east)+(0,-\vmtmpdist)$) edge node[below] {
        \begin{tabular}{@{}c@{}}
            \normalsize\texttt{:hasMember} \\
            \texttt{\textbf{joinDate}:\hspace*{.5ex}2012-07-21}
        \end{tabular}
    } ($(p2.west)+(0,-\vmtmpdist)$);
    
    \draw ($(f1.east)+(0,\vmtmpdist)$) edge[out=0,in=180] node[sloped, above] {\texttt{:containerOf}} (m2.west);
\end{tikzpicture}}
    \caption{Property graph instance based on the LDBC SNB Schema~\cite{DBLP:journals/corr/abs-2001-02299}.}
    \label{fig:running-example-graph}
\end{figure}

We sometimes consider \emph{unlabeled} graphs, i.e., graphs in which labels are indistinguishable. 
For technical reasons, we fix 
a special label~$\splbl\in\Lab$ and say that a graph $G=(V, E, \eta, \lambda, \pi)$ is \emph{unlabeled} if $\lambda(o) = \set{\splbl}$, for all $o \in V \cup E$. Additionally, we say that~$G$ is \emph{without data} if $G$ is unlabeled and~$\pi$ has an empty domain.

Several proofs will rely on a \emph{clique of a given size~$n$}, that is, a graph without data, with~$n$ vertices and, for every pair~$(s,t)$ of vertices, an edge~$e$ such that~$\eta(e)=(s,t)$; note that such cliques have self-loops. Additionally, some statements use the notions of \emph{subgraph} and \emph{induced-subgraph} relations.
A subgraph $G'$ of $G$ is simply a substructure of $G$; an induced subgraph necessarily retains all labels and properties of objects of $G$ that appear in $G'$. Both notions are formally defined below.

\begin{definition}[Subgraph, Induced Subgraph, adapted from~\cite{DBLP:books/daglib/0070576}]
    \label{d:subgraph}
    Let $G = (V, E, \eta, \lambda, \pi)$ be a graph. A \emph{subgraph} $G' = (V', E', \eta', \lambda', \pi')$, denoted by $G' \sqsubseteq G$, is a graph that is also a substructure of $G$, that is, it meets the following five conditions.
    \begin{gather*}
        V' \subseteq V \quad\qquad E' \subseteq \eta^{-1}(V'\times V') \quad\qquad \eta' = \eta\restriction_{E'}
        \quad\qquad
        \forall o\in V'\cup E'\quad\lambda'(o) \subseteq \lambda(o)\\
        \forall o\in V'\cup E', \forall a\in \key \quad (o,a)\in\dom(\pi') \implies \pi'(o,a) = \pi(o,a)
    \end{gather*}
    Moreover, we say that $G'$ is the subgraph of $G$ \emph{induced by} $V'$ and $E'$, denoted by~$G' \sqsubseteq_\text{ind} G$, if it satisfies:
    $\lambda' = \lambda\restriction_{V' \cup E'}$ and $\pi' = \pi\restriction_{(V' \cup E') \times \key}$.
\end{definition}

\subsection{Query Languages}
\label{sec:query-lang}
Our results are stated for \emph{conjunctive regular path queries}, as defined next. Our results also apply to the case of \emph{conjunctive queries}, which are here defined as a subcase of \crpq{}.

\begin{definition}[Conjunctive Regular Path Query]\label{d:crpq}
    A \emph{conjunctive regular path query} $(\crpq)$ is a tuple $Q(\bar{x}) = (V_Q, P_Q, \eta_Q, \lambda_Q)$ where $V_Q, P_Q \subset \var$ are pairwise disjoint finite sets of \emph{vertex} and \emph{path} variables, respectively; $\eta_Q: P_Q \to V_Q \times V_Q$ is a function that assigns \emph{source} and \emph{target} vertex variables to each path variable, $\lambda_Q = \lambda_{Q,1} \cup \lambda_{Q,2}$ is a labeling function with $\lambda_{Q,1}: V_Q \to {\lab}$ and $\lambda_{Q,2}: P_Q \to \reg(\lab)$; and $\bar{x}$ is a tuple of variables from $V_Q$ and $P_Q$, called \emph{output variables}.
    We write $\var_Q = V_Q \cup P_Q$ the set of variables that appear in $Q$. 
\end{definition}

Given two \crpqs{} $Q_1$ and $Q_2$, their conjunction, obtained by identifying identical variables, is also a $\crpq$, denoted as $Q_1 \land Q_2$. A \crpq{} is interpreted as a graph pattern, whose edges represent paths of the graph that conform to the given regular expression, as defined below.

\begin{definition}[Match]\label{d:homomorphism}
    Let $Q(\bar{x}) = (V_Q, P_Q, \eta_Q, \lambda_Q)$ be a $\crpq$ and $G = (V, E, \eta, \lambda, \pi)$ be a graph. A \emph{match} of $Q$ over $G$ is a function $h: V_Q \cup P_Q \to V \cup \paths(G)$ such that:
    \begin{enumerate}
        \item for all $u \in V_Q$, $h(u) \in V$ and $\lambda_Q(u)\in \lambda(h(u))$;
        \item for all $e \in P_Q$, $h(e) \in \paths(G)$ such that $\eta(h(e)) = h(\eta_Q(e))$ and $\lambda(h(e)) \cap L(\lambda_Q(e)) \neq \emptyset$.
    \end{enumerate}
    The \emph{answer} of $Q(x_1,\ldots,x_k)$ over $G$ is given by:
    \begin{equation}
        Q(G) =~\set{(h(x_1), \dots, h(x_k)) | h \text{ is a match of } Q \text{ over } G}
    \end{equation}
\end{definition}

\begin{remark}
    We only consider \emph{nonempty} paths.  Hence, we always assume, without loss of generality, that every regular expression~$r$ is such that~$L(r)$ does not contain the empty word.
\end{remark}

Given two \crpq{}s $Q_1(\bar{x_1}),Q_2(\bar{x_2})$ and two matches~$h_1,h_2$ of respectively $Q_1$ and $Q_2$, we say they are \emph{consistent} if $h_1(x)=h_2(x)$ for all~$x$ in both $\bar{x_1}$ and $\bar{x_2}$.

In the property graph model, both vertices and edges may bear properties, but paths are not allowed to.
Hence, conditions over property values will not be allowed over path variables unless they bind to a single edge; i.e., edge variables.

\begin{definition}[Edge variables, Conjunctive Queries] Let $Q(\bar{x}) \mathbin= (V_Q, P_Q, \eta_Q, \lambda_Q)$ be a \crpq.%
\begin{enumerate}
    \item
    A path variable~$x\in P_Q$ is called an \emph{edge variable} if~$\lambda_{Q}(x)$ describes a language containing only words of length~1 (that is,~$\lambda_{Q}(x)$ is a disjunction of labels).
    \item $Q$ is called a \emph{conjunctive query} (\cq) if all path variables in~$Q$ are edge variables.
\end{enumerate}
\end{definition}

In the special case of $\cq{}$, a match maps an edge (variable) of the query to a path of length 1 (that is, an edge).  A match is in fact a graph-to-graph homomorphism in that case.

In addition to the topological conditions expressed by \crpq{}, our queries can use (in)equalities between object identifiers, data values, and constants.

More precisely, \emph{predicates} we consider are expressions of the form $x = y$, $x.a = y.b$, $x.a = c$, $x \neq y$, $x.a \neq y.b$ or $x.a \neq c$, where $x, y \in \var$, $a, b \in \key$, and $c \in \val$.
The first three are \emph{equality} predicates and the others are \emph{inequality} predicates.

We also use the following derived predicates: $\bot$, for an arbitrary unsatisfiable set of predicates (for instance $\set{x.a = 1, x.a = 2}$) and $\ex(x.a)$ for~$x.a=x.a$ which ensures that the object bound to~$x$ has a property~$a$.
Similarly, we will use $\ex(\bar{\nu})$, for $\set{\nu=\nu | \nu\in\bar{\nu}}$, where~$\bar \nu$ is a tuple of expressions of the form $x$ or $x.a$, and $\ex(\bar{\nu}, \bar{\nu}')$ for $\ex(\bar{\nu}) \land \ex(\bar{\nu}')$.

\begin{definition}[Query]
    A \emph{query} is a pair $S = (Q(\bar{x}), C(\bar{x}))$ where~$Q$ is a \crpq{}, called the \emph{pattern} of the query and~$C$ is a set of predicates, called \emph{conditions}, such that for every occurrence of $x.a$ in~$C$, $x$ is a vertex or edge variable in~$Q$. We write $\var_S = \var_Q$.
    
    We define the conjunction of queries $S_1 = \big(Q_1(\bar{x}), C_1(\bar{x})\big)$ and $S_2 = \big(Q_2(\bar{x}'), C_2(\bar{x}')\big)$ as $S_1 \land S_2 = \big(Q_1 \land Q_2(\bar{x}, \bar{x}'), C_1 \cup C_2(\bar{x}, \bar{x}')\big)$.
    
    The set of all queries is denoted by \crpqneq{}. The sets
    \crpqeq{}, \crpqeqc{} and \crpqneqc{} are defined similarly, except predicates are restricted:~$C$ may contain only equality predicates for~\crpqeq{}, predicates of the form~$x.a=c$ for \crpqeqc{} and predicates of the form~$x.a=c$ and $x.a \neq c$ for \crpqneqc{}.
\end{definition}

A \emph{match} of a query $(Q,C)$ over $G$ is a match~$h$ of~$Q$ over~$G$ that satisfies all predicates in~$C$, after replacing $x$ and $x.a$ with $h(x)$ and the value associated with the property key $a$ of $h(x)$, respectively. See \cref{pred-satisfaction} in the appendix for details.

\begin{toappendix}
\begin{definition}[Satisfaction of predicates, conditions, and matches for patterns and queries]\label{pred-satisfaction}
    Let $G = (V,E,\eta, \lambda, \pi)$ be a graph and $h$ be a match of a pattern~$Q$ over~$G$. Let $p$ be a predicate. We say that $h$ \emph{satisfies} $p$, denoted as $h\models p$, if:
    \begin{itemize}
        \item If $p$ is of the form $x = y$ (resp. $x \neq y$), then $x,y\in\dom(h)$ and $h(x) = h(y)$ (resp. $h(x) \neq h(y)$);
        \item If $p$ is of the form $x.a = y.b$ (resp. $x.a \neq y.b$), then $x,y\in\dom(h)$, $(h(x),a), (h(y),b) \in \dom(\pi)$ and $\pi(h(x),a) = \pi(h(y),b)$ (resp. $\pi(h(x),a) \neq \pi(h(y),b)$);
        \item If $p$ is of the form $x.a = c$ (resp. $x.a \neq c$), then $x\in\dom(h)$, $(h(x),a) \in \dom(\pi)$ and $\pi(h(x),a) = c$ (resp. $\pi(h(x),a) \neq c$);
    \end{itemize}
    
    Finally, if $P$ is a set of predicates, we say that \emph{$h$ satisfies~$P$}, denoted by $h\models P$, if $h\models p$ for each $p\in P$. Given some pattern~$Q$ and conditions~$P$ we say that~$h$ is a \emph{match} of~$(Q,P)$ if it is a match of~$Q$ that satisfies~$P$.
\end{definition}
\end{toappendix}

\subsection{Revisiting Known Graph Constraint Languages}
\label{sec:constraint-lang}
We revisit the main graph constraint languages that form the basis of our study, namely \gfd{}~\cite{DBLP:conf/sigmod/FanWX16a}, \ggd{} \cite{DBLP:conf/cikm/ShimomuraFY20}, and \pgkeys{} \cite{DBLP:conf/sigmod/AnglesBDFHHLLLM21}, and recast them in a common, parametric framework.

\begin{definition}[\gfd{} and \ggd{}]
    \label{d:ggd}
    A \emph{Graph Generating Dependency} $(\ggd)$ for a query language~$\Ell$
    is a tuple $\varphi = (S(\bar{x}, \bar{y}) \Rightarrow T(\bar{x}, \bar{z}))$
    where $\bar{x},\bar{y},\bar{z}$ are disjoint tuples of variables, and
    $S$ and $T$ are queries from~$\Ell$. When patterns or conditions are trivial, i.e., when they do not restrict the topology or the data of the graph, we may omit them from the constraint.
    A \emph{Graph Functional Dependency} $(\gfd)$
    is a $\ggd${} where~$\bar{z}$ is empty and the pattern of $T$ is trivial. A $\gfd$ will be written as $(S(\bar{x}, \bar{y}) \Rightarrow C_t(\bar{x}))$ for short.

    Moreover, we say that a graph $G$ satisfies $\varphi$, denoted as $G \models \varphi$, if for all matches $h_s$ of $S$ over $G$ there exists a match $h_t$ of $T$ over $G$ that is consistent with $h_s$.
    Otherwise, we say that $G$ \emph{violates} $\varphi$, denoted as $G \not\models \varphi$.
\end{definition}

For instance, the \gfd{} $(\texttt{(x)}\xleftarrow{\texttt{:hasMember}} \texttt{(z)} \xrightarrow{\texttt{:hasMember}} \texttt{(y)} \Rightarrow \set{\texttt{x.speaks} = \texttt{y.speaks}})$ enforces that two members of a same forum speak the same language.
Similarly, the \ggd{} $(\texttt{(z)} \xrightarrow{\texttt{:containerOf}} \texttt{(y)} \xrightarrow{\texttt{:hasCreator}} \texttt{(x)} \Rightarrow \texttt{(z)} \xrightarrow{\texttt{:hasMember}}\texttt{(x)})$ enforces that the creator of a message is a member of the forum where the message is posted.

Graph Entity Dependencies (\ged, see~\cite{DBLP:journals/tods/FanL19})  are not explicitly defined in this document, as they are subsumed by our parametric definition of \gfd.
Indeed, the difference between \gfd{} and \ged{} lies in the fact that \ged{} can compare vertex and edge identities, while \gfd{} can only test data values. 
In our framework, this difference is captured by the choice of predicates that are allowed in the query language. Thus, the original definition of \gfd{} corresponds to disallowing equality between object identifiers.

\pgkeys{} \cite{DBLP:conf/sigmod/AnglesBDFHHLLLM21} have been introduced to identify unique vertices, edges, and properties in a property graph.
The \emph{scope} selects the objects constrained by the \pgkey.
The \emph{descriptor} contains a \emph{key expression}, specifying how to obtain a \emph{key value}, and an \emph{assertion keyword} specifying the kind of constraint.
A typical \pgkey{} looks as follows in the syntax of \cite{DBLP:conf/sigmod/AnglesBDFHHLLLM21}.
\begin{equation}\label{eq:def-pgk-orig}
    \FOR \overbrace{x \WITHIN Q_s(x, \bar{y}) \WHERE C_s(x, \bar{y})}^{\text{scope}} ~\underset{\begin{array}{r}\uparrow\\[-5pt]\llap{\text{\scriptsize assertion keyword}}\end{array}}{\alpha}~ \overbrace{\underset{\begin{array}{@{}l@{}}\uparrow\\[-5pt]\rlap{\text{\scriptsize key expression}}\end{array}}{\bar{\nu}}~ \WITHIN Q_t(x, \bar{z}) \WHERE C_t(x, \bar{z})}^{\text{descriptor}}
\end{equation}
For the sake of uniformity with \gfd{} and \ggd{}, we write \pgkeys{} as:
\begin{equation}\label{eq:def-pgk}
    \big(S(x, \bar{y}) \Rightarrow \alpha(\bar{\nu}), T(x, \bar{z})\big)
\end{equation}
where $\bar{y}$ and $\bar{z}$ are disjoint tuples of variables, $x \not\in \bar{y} \cup \bar{z}$,
$S = (Q_s, C_s)$ and $T = (Q_t, C_t)$ are queries, and $\bar{\nu}$ is a key expression over variables~$x \cup \bar{z}$.
As for \ggd, we may omit patterns or conditions if they are trivial.
The key expression~$\bar{\nu}$ is a tuple consisting of elements of the form~$x$ or~$x.a$ with~$x\in\var$ and~$a\in\key$.
The assertion keyword~$\alpha$ is either $\MANDATORY$, $\SINGLETON$ or $\EXCLUSIVE$. Intuitively, it enforces that no two objects matched by $x$ can share the same key value ($\EXCLUSIVE$), or that all objects matched by $x$ have at least ($\MANDATORY$) or at most ($\SINGLETON$) one key value.
Satisfaction is defined as follows.

\begin{definition}[\pgkey{} Satisfaction]\label{d:pgk}
    Given a \pgkey{} $\varphi = \big(S(x,\bar{y}) \Rightarrow \alpha(\bar{\nu}), T(x,\bar{z})\big)$, we say that a graph~$G$ \emph{satisfies} $\varphi$ denoted $G \models \varphi$, if the conditions below on $\alpha$ hold:
    \begin{itemize}
        \item $\alpha =~ \MANDATORY$: for all matches~$h_s$ of~$S$, there exists a consistent match~$h_t$ of~$T$ such that~$h_t(\bar{\nu})$ is defined.
        \item $\alpha =~ \SINGLETON$: for all matches~$h_s$ of~$S$, and~$h_{t},h'_{t}$ of~$T$ consistent with~$h_s$, it holds that $h_{t}(\bar{\nu}) = h'_{t}(\bar{\nu})$.
        \item $\alpha =~ \EXCLUSIVE$: for all matches~$h_{s},h'_{s}$ of~$S$, and~$h_{t},h'_{t}$ of~$T$ respectively consistent with~$h_{s},h'_{s}$, and such that~$h_{t}(\bar{\nu}) = h'_{t}(\bar{\nu})$, it holds that $h_{s}(x) = h'_{s}(x)$.
    \end{itemize}
    Otherwise we say that~$G$ \emph{violates} $\varphi$, denoted as $G \not\models \varphi$.
\end{definition}

We illustrate \cref{d:pgk} with the following \pgkeys{}: \eqref{eq:ex:mandatory} 
enforces that all persons that have posted a message have an email address;
\eqref{eq:ex:singleton} enforces that a message has at most one creator;
\eqref{eq:ex:exclusive} enforces that no two distinct persons can have the same email.
\begin{gather}
    \label{eq:ex:mandatory}
    \big(x\texttt{:Forum} \xrightarrow{\texttt{:containerOf}} y\texttt{:Message} \xrightarrow{\texttt{:hasCreator}} z\texttt{:Person} \Rightarrow~ \MANDATORY(z.\texttt{email})\big)\\
    \label{eq:ex:singleton}
    \big(x\texttt{:Message} \Rightarrow~ \SINGLETON(y), x \xrightarrow{\texttt{:hasCreator}} y\big) \\
    \label{eq:ex:exclusive}
    \big(x\texttt{:Person} \Rightarrow~ \EXCLUSIVE(x.\texttt{email})\big)
\end{gather}
Aside from the difference in syntax (\eqref{eq:def-pgk-orig} vs \eqref{eq:def-pgk}), our presentation of \pgkeys{} differs slightly from the original definition in \cite{DBLP:conf/sigmod/AnglesBDFHHLLLM21} with respect to syntax restrictions.
This facilitates comparisons with other constraint languages and does not affect the expressive power. In \cite{DBLP:conf/sigmod/AnglesBDFHHLLLM21} a single \pgkey{} may have multiple assertion keywords, among which $\EXCLUSIVE$ must appear. First, making the $\EXCLUSIVE$ keyword mandatory does not affect the expressive power. Indeed, the \pgkey{} {$\big(S(x, \bar{y}) \Rightarrow~ \SINGLETON(\bar{\nu}), T(x, \bar{z})\big)$}, forbidden in \cite{DBLP:conf/sigmod/AnglesBDFHHLLLM21}, can be equivalently written as $\big(S(x, \bar{y}) \Rightarrow~ \highlight{\EXCLUSIVE\,}\SINGLETON(\highlight{x}, \bar{\nu}), T(x, \bar{z})\big)$, where the $\EXCLUSIVE$ keyword has no effect since the key expression now contains $x$ itself.
The same trick also works with $\MANDATORY$ instead of $\SINGLETON$.

Second, a \pgkey{} $\varphi$ with two assertion keywords $\alpha$ and $\beta$ can always be split into two identical $\pgkeys$ $\varphi_\alpha$ and $\varphi_\beta$, except that $\varphi_\alpha$ (resp.~$\varphi_\beta$) only contains the keyword $\alpha$ (resp.~$\beta$), which conforms to~\eqref{eq:def-pgk}. See \cref{a:p:split-pgk} in \cref{a:split-pgk-appendix} for more details.

\begin{remark}\label{unconstrained}
As can be seen in Equation~\ref{eq:def-pgk}, the scope and the descriptor of \pgkey{} share exactly one variable. Lifting this restriction leads to what we call \emph{unconstrained \pgkeys}, equivalent to GGD (Remark~\ref{r:npgk-is-ggd-cq-eq}) and the fragment targeted by the translations in \cite{DBLP:journals/is/ManouvrierB26}. 
\end{remark}

\begin{toappendix}
\label{a:split-pgk-appendix}
    \begin{property}
    \label{a:p:split-pgk}
        Let $\varphi = \big(S(x, \bar{y}) \Rightarrow \alpha(\bar{\nu}), T(x, \bar{z})\big)$ be a \pgkey{} over any query language. We assume~$\alpha \neq~ \EXCLUSIVE$. Then~$\varphi$ is equivalent to the set of \pgkeys{} $\{\psi_1,\psi_2\}$, given below.
        \begin{align}
            \psi_1 ={} & \big(S(x, \bar{y}) \Rightarrow \alpha(\bar{\nu}, x), T(x, \bar{z})\big) \\
            \psi_2 ={} & \big(S(x, \bar{y}) \Rightarrow~ \EXCLUSIVE(\bar{\nu}, x), T(x, \bar{z})\big)
        \end{align}
    \end{property}
\end{toappendix}

\section{Interesting Fragments and First Results}\label{sec:gen}
As can be seen in \eqref{eq:def-pgk}, the scope and descriptor of a \pgkey{} always share a single variable, whereas \gfd{} and \ggd{} do not have this restriction. In order to make a finer comparison of the expressiveness of the three languages, we define subclasses of \gfd{} and \ggd{} based on how many variables are shared between the left and the right sides. We will show that this parameter has an impact on the expressiveness of these constraint languages.
It will also be useful to consider $\mpgkeys$, the subclass of~$\pgkeys$ in which all constraints use the assertion keyword $\MANDATORY$.

\begin{definition}[$n\gfd$ and $n\ggd$]\label{d:ngfd}\label{d:nggd}
A $\gfd$ $(S(\bar{x}, \bar{y}) \Rightarrow C_t(\bar{x}))$ is in $n\gfd$ if $|\bar{x}|\leq n$.
Similarly, a $\ggd$ $(S(\bar{x}, \bar{y}) \Rightarrow T(\bar{x}, \bar{z}))$ is in $n\ggd$ if $|\bar{x}|\leq n$.
\end{definition}

For instance, $(\texttt{(x)} \xrightarrow{\texttt{:hasModerator}} \texttt{(y)}, \set{\texttt{x.name} = \texttt{Foo}} \Rightarrow \texttt{(z)} \xrightarrow{\texttt{:hasCreator}} \texttt{(y)})$  enforces that moderators of forum Foo must have already posted on some forum. The two sides of the implication have a single shared variable \texttt{y}, hence~$\varphi$ is a~$1\ggd{}$.
If one wants to enforce that the message is in the same forum, one could use the following $2\ggd$.
\begin{equation*}
    (\texttt{(x)} \xrightarrow{\texttt{:hasModerator}} \texttt{(y)}, \set{\texttt{x.name} = \texttt{Foo}} \Rightarrow \texttt{(y)} \xleftarrow{\texttt{:hasCreator}} \texttt{(z)} \xleftarrow{\texttt{:containerOf}} \texttt{(x)})
\end{equation*}
Note that, by sharing a single edge variable \texttt{e}, one could define the equivalent 1\ggd{} below.
\begin{equation*}
    (\texttt{(x)} \xrightarrow{\texttt{e:hasModerator}} \texttt{(y)}, \set{\texttt{x.name} = \texttt{Foo}} \Rightarrow \texttt{(y')} \xleftarrow{\texttt{:hasCreator}} \texttt{(z')} \xleftarrow{\texttt{:containerOf}} \texttt{(x')} \xrightarrow{\texttt{e}} \texttt{(y')})
\end{equation*}

First, let us show that the $n\gfd$ hierarchy collapses at the second level for equality and inequality predicates.

\begin{lemma}\label{l:split-gfd}
    $\gfd$ and $2\gfd$ are equivalent for $\cqeq$, $\cqneq$, $\crpqeq$, and $\crpqneq$.
\end{lemma}
\begin{proof}[Proof Sketch]
    Any \gfd{}~$\varphi = (S \Rightarrow \set{d_1, \dots, d_t})$ 
    is equivalent to the following set of $2\gfds$~$\set{(S \Rightarrow \set{d_1}), \dots, (S \Rightarrow \set{d_t})}$.
\end{proof}

Next, we establish that, unlike $n\gfd$, the $n\ggd$ hierarchy is strict: given a number~$n$ of variables, we will show that we can construct two graphs that differ only by the presence or the absence of a common out-neighbor for every vertices of a given cycle whose size depends on~$n$. This feature turns out to require the information of~$n$ variables to be remembered, i.e. both graphs are indistinguishable if we forget a variable.

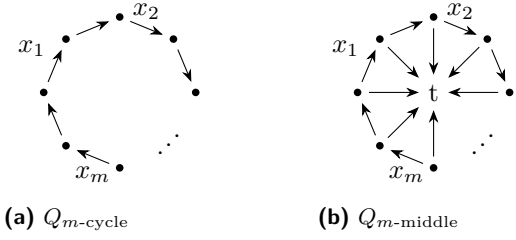
\begin{figure}[t]
    \centering
    \subfloat[$Q_{m\text{-cycle}}$]{%
        \label{fig:cq-m-cycle}%
        \begin{tikzpicture}[
    dot/.style={fill=black,circle,minimum width=1mm, inner sep=0pt, outer sep=1mm},
    every edge/.style = {
        draw = black,
        ->,
    }]

    \path (0,0)++(180:12.5mm) node[dot] (z1) {};
    \path (0,0)++(135:12.5mm) node[dot] (z2) {};
    \path (0,0)++(90:12.5mm) node[dot] (z3) {};
    \path (0,0)++(45:12.5mm) node[dot] (z4) {};
    \path (0,0)++(0:12.5mm) node[dot] (z5) {};
    \path (0,0)++(-90:12.5mm) node[dot] (z7) {};
    \path (0,0)++(-135:12.5mm) node[dot] (z8) {};

    \path[draw,->,red] (z1) to node[left] {$x_1$} (z2);
    \path[draw,->] (z2) to (z3);
    \path[draw,->,red] (z3) to node[above] {$x_2$} (z4);
    \path[draw,->] (z4) to (z5);
    \path[bend left] (z5) to node[sloped] {$\cdots$} (z7);
    \path[draw,->,red] (z7) to node[below] (xm) {$x_m$} (z8);
    \path[draw,->] (z8) to (z1);
\end{tikzpicture}}%
    \hspace*{15mm}%
    \subfloat[$Q_{m\text{-middle}}$]{%
        \label{fig:cq-m-middle}%
        \begin{tikzpicture}[
    dot/.style={fill=black,circle,minimum width=1mm, inner sep=0pt, outer sep=1mm},
    every edge/.style = {
        draw = black,
        ->,
    }]

    \path (0,0)++(180:12.5mm) node[dot] (z1) {};
    \path (0,0)++(135:12.5mm) node[dot] (z2) {};
    \path (0,0)++(90:12.5mm) node[dot] (z3) {};
    \path (0,0)++(45:12.5mm) node[dot] (z4) {};
    \path (0,0)++(0:12.5mm) node[dot] (z5) {};
    \path (0,0)++(-90:12.5mm) node[dot] (z7) {};
    \path (0,0)++(-135:12.5mm) node[dot] (z8) {};
    \node (t) at (0,0) {t};

    \path[draw,->,red] (z1) to node[left] {$x_1$} (z2);
    \path[draw,->] (z2) to (z3);
    \path[draw,->,red] (z3) to node[above] {$x_2$} (z4);
    \path[draw,->] (z4) to (z5);
    \path[bend left] (z5) to node[sloped] {$\cdots$} (z7);
    \path[draw,->,red] (z7) to node[below] (xm) {$x_m$} (z8);
    \path[draw,->] (z8) to (z1);
    \foreach \i in {1,2,3,4,5,7,8} {
        \path[draw,->] (z\i) to (t);
    }
\end{tikzpicture}}
     \caption{Queries for the Proof of \cref{t:nggd-sub-mggd}.}
     \label{fig:cq-m-cycle'}
\end{figure}

\begin{theorem}
    \label{t:nggd-sub-mggd}
    Let $n<m$. Then $n\ggd \subsetneq m\ggd$ for~$\cqeq$,~$\cqneq$, $\crpqeq$, and~$\crpqneq$.
\end{theorem}

\begin{proof}
    
    Let $n<m$, and $Q_{m\text{-cycle}}$ and $Q_{m\text{-middle}}$ be the $\cq$s represented in \cref{fig:cq-m-cycle'}.
    $Q_{m\text{-cycle}}$ is a cycle of length~$2m$ in which every other edge (in red in the figure) bears an output variable ($x_1$ to $x_m$).
    $Q_{m\text{-middle}}$ is a copy of $Q_{m\text{-cycle}}$ augmented with an output vertex variable~$t$, which is the target of an edge from every vertex of the cycle.
    Let $\varphi$ be the $m\ggd$ $\varphi = (Q_{m\text{-cycle}} \Rightarrow Q_{m\text{-middle}})$. We show that no set of $n\ggd$ $\Psi$ can express $\varphi$. 
    To do so, we proceed by contradiction.
    From~$\Psi$, we construct two graphs $G_1$ and $G_2$ such that $G_1 \not \models \varphi$ and $G_2 \models \varphi$.
    We then show $G_2 \models \Psi \implies G_1 \models \Psi$, that is, $\Psi$ cannot accept $G_2$ without also accepting $G_1$.

    \underline{\bf Construction of $G_1$ and $G_2$.} Let $k = 2m + \max\limits_{\psi \in \Psi}(|\psi|)+2$. Let $S_1,\ldots,S_{2m}$ be disjoint sets of vertices of size $k$. Let $v_1,\ldots,v_{2m}$ be vertices that do not occur in $S_1,\ldots,S_{2m}$. $G_1$ is a graph without data, with vertices $S_1\cup\ldots\cup S_{2m}\cup\set{v_1,\ldots,v_{2m}}$ and edges such that:
    \begin{itemize}
        \item $S_1\cup\ldots\cup S_{2m}$ is a clique without self-loops.
        \item For each~$i\in [2m]$, $S_i\cup \set{v_j \mid j\neq i}$ is a clique without self-loops.
        \item We add an edge~$e_i$ for each $i < 2m$ defined by $v_i\xrightarrow{e_i} v_{i+1}$; and one edge $v_{2m} \xrightarrow{e_{2m}} v_1$.
    \end{itemize}
    Remark that, as intended, $G_1 \not\models \varphi$. Indeed, $h : x_i \mapsto e_{2i}$ is a match of $Q_{m\text{-cycle}}$, but there is no match of $Q_{m\text{-middle}}$ that is consistent with $h$. Indeed, there is no suitable choice for $t$, as vertices in $S_i$ lack an edge to $v_i$ and vertices $v_i$ lack a self-loop.

    $G_2$ is defined as a copy of $G_1$ with an additional fresh vertex $\omega$ and edges $v_i \rightarrow \omega$ for all $i \in [2m]$.
    Notice that $G_2 \models \varphi$.
    Indeed, let $h$ be a match of $Q_{m\text{-cycle}}$. First, remark that no edge adjacent to $\omega$ can be in the image of $h$, as those edges are part of no cycles.
    Let~$V$ be the set of the endpoints of the edges in $h(\set{x_1,\ldots,x_m})$.
    If~$V = \set{v_1,\ldots,v_{2m}}$, then the consistent match $h'$ such that $h'(t) = \omega$ is a suitable choice for $Q_{m\text{-middle}}$.
    Otherwise, there exists $j$ such that $v_j\notin V$. Then, let $v$ be any vertex in $S_j \setminus V$; it exists because $S_j$ is of size $k > 2m$.
    Note that $v$ is a neighbor of every vertex in $G_2$ but $v_j$ and $\omega$,
    hence the consistent match $h'$ such that $h'(t) = v$ is a suitable choice for $Q_{m\text{-middle}}$. 
    
    \underline{\bf Proof that $G_2\models \Psi \implies G_1 \models \Psi$.} 
    We now assume that~$\Psi$ accepts~$G_2$ and show that it then also accepts~$G_1$.
    Let $\psi = (S(\bar{x}, \bar{y}) \Rightarrow T(\bar{x}, \bar{z}))$ be an $n\ggd$ in $\Psi$; in particular $G_2 \models \psi$.
    We write~$\bar{x} = (x_1,\ldots,x_n)$.

    If there is no match of $S$ over $G_1$, then $G_1 \models \psi$ trivially holds.
    Henceforth, we assume there is a match $h$ of $S$ over $G_1$. Then $h$ is also a match of $S$ over $G_2$ and since~$G_2\models\psi$, $T$ is satisfied over $G_2$ by some match $h_2$ that is consistent with $h$. 

    We let~$V$ denote the set of vertices (of~$G_1$) containing all vertices in~$h(\{x_1,\ldots,x_n\})$ and all vertices that are adjacent to an edge in~$h(\{x_1,\ldots,x_n\})$.
    Note that~$|V|<2m$ hence there is some~$j$ such that~$v_j\notin V$.
    Let~$\omega'$ and~$v_j'$ be two distinct vertices in~$S_j$ that are not in~$\im(h_2)$;
    two such vertices can be found due to the definition of~$k$.
    \begin{enumerate}
        \item For each vertex variable~$\xi\in \var_T$ such that~$h_2(\xi)\notin\{\omega,v_j\}$,
        we set~$h_1(\xi)=h_2(\xi)$.
        \item For each vertex variable~$\xi\in \var_T$ such that~$h_2(\xi)=\omega$,
        we set~$h_1(\xi)=\omega'$.
        \item For each vertex variable~$\xi\in \var_T$ such that~$h_2(\xi)=v_j$,
        we set~$h_1(\xi)=v_j'$.
        \item For each path variable~$e\in \var_T$ from $\xi_s$ to $\xi_t$, $h_2(e)$ is a path from $h_2(\xi_s)$ to $h_2(\xi_t)$ in $G_2$. Let $u_1$ be the first vertex of $h_2(e)$ and $u_2,u_3$ be the last two vertices of $h_2(e)$. We construct $h_1(e)$ as a copy of $h_2(e)$ with possible replacements for $u_1$, $u_2$ and $u_3$. First, remark that if $\omega$ occurs in $h_2(e)$, then $u_3 = \omega$, as $\omega$ has no outgoing edges. Then:
        \begin{itemize}
            \item If $u_2 \neq v_j$ and $u_3 = \omega$, we replace the last vertex with $\omega'$. This produces a valid path of $G_1$, as $\omega'$ is adjacent to every vertex of $G_1$ but $v_j$.
        
            \item If $u_2 = v_j$ and $u_3 = \omega$, we replace these vertices with $v'_j$ and $\omega'$ respectively. This produces a valid path of $G_1$, as $v'_j$ is adjacent to every vertex of $G_1$ but $v_j$. Specifically, since $v_j$ has no self-loop, the possible predecessor of $u_2$ in $h_2(e)$ is not $v_j$.
        
            \item If $u_1 = v_j$, we replace the first vertex with $v_j'$. Again, this is well defined, as the successor of $u_1$ cannot be $v_j$. The case where the successor of $u_1$ is $\omega$ corresponds to a path of length 1 and is thus included in the previous case.
        \end{itemize}
        Note that this procedure preserves inequalities of paths: $h_2(e) \neq h_2(e')$ if and only if $h_1(e) \neq h_1(e')$. Indeed, since $\omega'$ and $v_j'$ are not in $\im(h_2)$, no path in $\im(h_2)$ starts with $u_1 = v_j'$ nor ends with $u_3 = \omega'$. Similarly, no path in $\im(h_2)$ ends with $u_2 = v_j'$ and $u_3 = \omega$, as there is no edge from $v_j'$ to $\omega$.
    \end{enumerate}
    It can be verified that all conditions in $T$ are satisfied by~$h_1$ because they are satisfied by~$h_2$: there is no data and identifier equality/inequality constraints are preserved.
\end{proof}

We conclude Section~\thesection{} by showing that $\gfd$ is distinct from the other classes, since it is the only one that is closed under induced subgraphs.

\begin{propositionrep}[Closure under induced subgraph]
    \label{p:closure-induced-subgraph}
    The following statements hold for~$\cqeq$, $\cqneq$, $\crpqeq$, and~$\crpqneq$.
    \begin{enumerate}
        \item \label{p:closure-induced-subgraph:gfd}
        $\gfd$ is closed under induced subgraph: let $\varphi$ be a $\gfd$ and $G$ be a graph. If $G \models \varphi$, then for all $G' \sqsubseteq_\text{ind} G$, $G' \models \varphi$.
    \item  \label{p:closure-induced-subgraph:exist} $1\ggd$ (resp.~\mpgkeys{}) is \emph{not} closed under induced subgraph: there exists a $1\ggd$ (resp.~\mpgkey{}) $\psi$ and two graphs $G, G'$ such that $G' \sqsubseteq_\text{ind} G$, $G \models \psi$ and $G' \not\models \psi$.
    \item \label{p:closure-induced-subgraph:sep} $\gfd \neq \ggd$, $1\ggd \not\subseteq \gfd$, $\mpgkeys \not\subseteq \gfd$, and $\gfd \neq \pgkeys$
    \end{enumerate}
\end{propositionrep}
\begin{proof}[Proof of 1.]
    Let $\varphi = (S(\bar{x}, \bar{y}) \Rightarrow C_t(\bar{x}))$ be a $\gfd$, and $G, G'$ be graphs such that $G \models \varphi$ and $G' \sqsubseteq_\text{ind} G$.
    
    Let us suppose that there exists a match $h$ of $S$ over $G$.
    If there does not exist a match $h'$ of the pattern of $S$ over $G'$, then there is nothing more to verify, and if there exists such a match, then $h' \models S$ and $h' \models C_t$ by definition of $\lambda'$ and $\pi'$.
    
    Now, let us suppose that there exists no match of $S$ over $G$. If there exists a match of the pattern of $S$ over $G$ that does not satisfy the conditions of $S$, then there exists no match of $S$ over $G'$ for the same reason, and if there exists no match of the pattern of $S$ 
    over $G$, then there exists no match of that pattern over $G'$ by monotonicity of~$\cqs$ and~$\crpqs$.
    In all cases: $G' \models \varphi$.
\end{proof}
\begin{appendixproof}[Proof of 2.]
    Let $\psi = (Q_s(\bar{x}, \bar{y}), \emptyset \Rightarrow Q_t(\bar{x}, \bar{z}), \emptyset)$ be a $\ggd$,
    and $G$ be a graph with exactly two connected components $\Gamma_s$ and $\Gamma_t$
    such that there exists a match $h_s$ of $Q_s$ over $\Gamma_s$ and a match $h_t$ of $Q_t$ over $\Gamma_t$, but there exists no match of $Q_t$ over $\Gamma_s$.
    If $G' = \Gamma_s$, then $G'$ is indeed an induced subgraph of $G$, but $G' \not\models \psi$, while $G \models \psi$.
    
    We prove the statement similarly for $\mpgkeys$, which immediately yields the result for $\pgkeys$.
\end{appendixproof}
\begin{toappendix}
    Item \enumstyle{3} follows from items \enumstyle{1} and \enumstyle{2}.
\end{toappendix}

\section{Queries with Equality}\label{sec:cq:eq}
In this section, we study $\crpqeq$. The results are summarized in \cref{thm:sec4} and \cref{fig:venn-eq}, with all inclusions shown in the figure being strict. Section~\ref{ssec:inclusions-eq} establishes the inclusions via direct translations, while Section~\ref{ssec:separations-eq} proves their strictness.

\subsection{Inclusions}
\label{ssec:inclusions-eq}

Some inclusions follow directly from syntactic restrictions, such as $\gfd \subseteq \ggd$, $1\ggd \subseteq \ggd$, and $\mpgkeys \subseteq \pgkeys$.
Moreover, $1\ggd$ and \mpgkeys{} are easy to translate into one another, as stated next.

\begin{propositionrep}
    \label{p:mpgk-is-1ggd-cq-eq}
    For~$\crpqeq$, $\mpgkeys$ and $1\ggd$ are equivalent.
\end{propositionrep}
\begin{proof}
    The proof is a simple syntactic rewriting of the $1\ggd$ into an $\mpgkey$.
    
    Let $\varphi = (S(x, \bar{y}) \Rightarrow T(x, \bar{z}))$ be a $1\ggd$.
    Then $\varphi$ is equivalent to the $\mpgkey$~$\psi$, below.
    \begin{equation}
        \psi =~ \big(S(x, \bar{y}) \Rightarrow~ \MANDATORY(x, \bar{z}), T(x, \bar{z})\big)
    \end{equation}

    Indeed, $\varphi$ and $\psi$ both enforce that, if there is a match of $S$, then there is a match of $T$. Conversely, if $\psi$ is the following \pgkey:
    \begin{equation*}
        \psi =~ \big(S(x, \bar{y}) \Rightarrow~ \MANDATORY(\bar{\nu}), T(x, \bar{z})\big)
    \end{equation*}

    Then $\psi$ is equivalent to $\varphi = \big(S(x, \bar{y}) \Rightarrow T(x, \bar{z}) \land \ex(\bar{\nu})\big)$.
\end{proof}

It follows immediately from \cref{p:mpgk-is-1ggd-cq-eq} that $1\ggd \subseteq \pgkeys$.
Furthermore, $\ggd$ with more shared variables can simulate the $\SINGLETON$ and $\EXCLUSIVE$ keywords by using two disjoint copies of the scope and descriptor of the \pgkey{}.

\begin{proposition}
    \label{p:pgk-sub-ggd-cq-eq}
    For~$\crpqeq$, we have that $\pgkeys \subseteq \ggd$.
\end{proposition}
\begin{proof}
    Let $\varphi = \big(S(x, \bar{y}) \Rightarrow \alpha(\bar{\nu}), T(x, \bar{z})\big)$ be a \pgkey.
    We denote by $\bar{\nu}'$ the copy of $\bar{\nu}$ in which each variable $t$ that occurs in $\bar{\nu}$ is replaced by $t'$.
    Then $\varphi$ is equivalent to one of the following $\ggd$, depending on the assertion keyword $\alpha$:
    \begin{align*}
        \EXCLUSIVE:{}& \big(S(x, \bar{y}) \land S(x', \bar{y}') \land T(x, \bar{z}) \land T(x', \bar{z}') \land \set{\bar{\nu} = \bar{\nu}'} \Rightarrow \set{x = x'}\big) \\
        \SINGLETON:{}& \big(S(x, \bar{y}) \land T(x, \bar{z}) \land T(x, \bar{z}') \land \ex(\bar{\nu}, \bar{\nu}') \Rightarrow \set{\bar{\nu} = \bar{\nu}'}\big) \\
        \MANDATORY:{}& \big(S(x, \bar{y}) \Rightarrow T(x, \bar{z}) \land \ex(\bar{\nu})\big)
    \end{align*}
    
    It remains to remark that these $\ggds$ do indeed simulate the behavior of assertion keywords as defined in \cref{d:pgk}.
    For instance, the formula for $\EXCLUSIVE$ enforces that when two matches agree on the key expression ($\bar \nu = \bar \nu'$), then they must also agree on the variable in the scope ($x = x'$).
\end{proof}

\begin{remark}
    The translations of $\EXCLUSIVE$ and $\SINGLETON$ are actually in $\gfd$; the translation of $\MANDATORY$ is a $1$\ggd{}. Thus, we proved more precisely that $\pgkeys \subseteq 2\ggd$.
\end{remark}

If one allows equality only with a constant ($\doteqc$), an adaptation of \cref{l:split-gfd} shows that $\gfd=1\gfd$.
The next statement follows.

\begin{proposition}
    \label{p:gfdc-sub-ggd-cq-eq}
    $\gfd$ for $\crpqeqc$ are expressible in $1\ggd$ for $\crpqeq$.
\end{proposition}

Finally, we establish the most surprising result of Section~\thesection{}: despite having only one shared variable, $\pgkeys$ can simulate $\gfd$ through a clever use of the $\SINGLETON$ keyword.

\begin{propositionrep}\label{p:gfd-sub-pgk-cq-eq}
    For~$\crpqeq$, we have that $\gfd\subseteq\pgkeys$.
\end{propositionrep}
\begin{proofsketch}
    Let~$\varphi$ be a \gfd. Using the same argument as in the proof of \cref{l:split-gfd}, we may assume that $\varphi$ has the form $(S(x, y, \bar{z}) \Rightarrow C_t(x,y))$ with $|C_t| = 1$. We consider only the case where $C_t = \set{x.a = y.b}$. Other cases are similar; see \cref{a:p:gfd-sub-pgk-cq-eq}.

    Let~$x'$ and~$y'$ be fresh variables, and~$\bar{z}'$ be a tuple of fresh variables with~$|\bar{z}'|=|\bar{z}|$.
    Then~$\varphi$ is equivalent to the set of \pgkeys{} $\set{\psi_1, \psi_2, \psi_3}$ defined as follows.
    \begin{align}
        \psi_1 ={} & \big(S(\highlight{x}, y', \bar{z}') \Rightarrow \highlight{\MANDATORY(x)}, S(x, y, \bar{z}) \highlight{\land \set{x.a = y.b}}\big) \\
        \psi_2 ={} & \big(S(x', \highlight{y}, \bar{z}') \Rightarrow \highlight{\MANDATORY(y.b)}\big) \\
        \psi_3 ={} & \big(S(\highlight{x}, y', \bar{z}') \Rightarrow \highlight{\SINGLETON(y.b)}, S(x, y, \bar{z})\big)
    \end{align}
    For readability, we highlight the differences between the three formulas in gray.
    For all matches~$h$ of~$S$:
    \begin{itemize}
        \item $\psi_1$ ensures that there exists a match $h'$ of~$S$ such that $h(x) = h'(x)$ and $h(x.a) = h'(y.b)$, that is, each $x.a$ is equal to \emph{some} suitable $y.b$;
        \item $\psi_2$ ensures that $h(y)$ has a $b$ key;
        \item $\psi_3$ ensures that all matches $h'$ of~$S$ such that $h(x) = h'(x)$ agree on $h'(y.b)$ if $h'(y)$ has a $b$ key.
    \end{itemize} 
    It remains to remark that, together, these three formulas exactly enforce $\varphi$.
\end{proofsketch}
\begin{proof}
    \label{a:p:gfd-sub-pgk-cq-eq}
    Let~$\varphi$ be a \gfd. W.l.o.g we may assume that $\varphi$ has the form
    $(S(x, y, \bar{z}) \Rightarrow C_t(x,y))$, i.e. it is a $2\gfd$, from \cref{l:split-gfd}, and that~$|C_t|=1$, otherwise it may be split into one $\gfd$ per item in $C_t$.
    We aim to construct a set~$\Psi$ of \pgkeys{} that is equivalent to $\varphi$.
    Note that~$C_t$ has one of the following forms: $\set{x.a = y.b}$, $\set{x.a = c}$, or~$\set{x = y}$.
    
    Case where~$\varphi = (S(x, y, \bar{z}) \Rightarrow \set{x.a = y.b})$.
    Let~$x'$ and~$y'$ be fresh variables, and~$\bar{z}'$ be a tuple of fresh variables with~$|\bar{z}'|=|\bar{z}|$.
    The set $\Psi$ must ensure that~$h(x.a)$ and~$h(y.b)$ are defined for all matches~$h$ of~$S$, and that for all sets of matches~$\set{h_1, \ldots, h_n}$ of~$S$ such that~$h_i(x) = h_j(x)$ (respectively~$h_i(y) = h_j(y)$), $h_i(y.b) = h_j(y.b)$ (respectively~$h_i(x.a) = h_j(x.a)$) for~$i,j \in [n]$.
    Let~$\Psi = \set{\psi_1, \psi_2, \psi_3}$ be defined as follows.
    \begin{align}
        \psi_1 ={} & \big(S(\highlight{x}, y', \bar{z}') \Rightarrow \highlight{\MANDATORY(x)}, S(x, y, \bar{z}) \highlight{\land \set{x.a = y.b}}\big) \\
        \psi_2 ={} & \big(S(x', \highlight{y}, \bar{z}') \Rightarrow \highlight{\MANDATORY(y.b)}\big) \\
        \psi_3 ={} & \big(S(\highlight{x}, y', \bar{z}') \Rightarrow \highlight{\SINGLETON(y.b)}, S(x, y, \bar{z})\big)
    \end{align}
    For the sake of readability, the differences between each formula are highlighted in gray.

    Let~$G = (V, E, \eta, \lambda, \pi)$ be a property graph.
    Let~$\theta$ be the binary relation over~$V\cup E$ defined by:~$e \mathrel{\theta} f$ if there exists a match~$h$ of $S$ such that~$h(x)=e$ and~$h(y)=f$.
    We then have the following equivalences.
    \begin{itemize}
        \item $G\models\psi_1$ is equivalent to: for all~$e\in\dom(\theta)$, $\pi(e,a)$ exists and there is~$f\in (e\theta)$ such that $\pi(f,b)$ exists and~$\pi(e,a)=\pi(f,b)$.
        \item $G\models \psi_2$ is equivalent to: for all~$f\in\im(\theta)$, $\pi(f,b)$ exists.
        \item $G\models\psi_3$ is equivalent to: for all~$e\in\dom(\theta)$ and $f,f'\in (e\theta)$ such that $\pi(f,b)$ and $\pi(f',b)$ exists,
        we have that~$\pi(f,b)=\pi(f',b)$.
    \end{itemize}
    Hence~$G$ satisfies~$\{\psi_1,\psi_2,\psi_3\}$ if and only if for all~$e,f$ such that~$e \mathrel{\theta} f$: $\pi(e,a)$ exists, $\pi(f,b)$ exists, and~$\pi(e,a)=\pi(f,b)$.
    It is precisely what~$G\models \varphi$ means.
 
    The case~$\varphi = (S(x, \bar{z}) \Rightarrow \set{x.a = c})$ is similar but easier. It is equivalent to the following \pgkey{}.
    \begin{equation}\label{eqn:gfdc->mpgk}
        \big(S(x, \bar{z}) \Rightarrow~ \MANDATORY(x), \set{x.a = c}\big)
    \end{equation}

    Similarly, the case~$\varphi = (S(x, y, \bar{z}) \Rightarrow \set{x = y})$ is equivalent to the set of \pgkeys{} $\set{\psi_4, \psi_5}$ where $\psi_4$ and $\psi_5$ are defined as follows.
    \begin{align}
        \psi_4 ={} & \big(S(x, y', \bar{z}') \Rightarrow~ \MANDATORY(x), S(x, y, \bar{z}) \land \set{x = y}\big) \\
        \psi_5 ={} & \big(S(x, y', \bar{z}') \Rightarrow~ \SINGLETON(y), S(x, y, \bar{z})\big)
    \end{align}
\end{proof}

For instance, the \gfd{} $(S(\texttt{x},\texttt{y},\texttt{z}) \Rightarrow \set{\texttt{x.speaks} = \texttt{y.speaks}})$, where $S(\texttt{x},\texttt{y},\texttt{z})$ contains only the pattern $(\texttt{(x)} \xleftarrow{\texttt{:hasMember}} \texttt{(z)} \xrightarrow{\texttt{:hasMember}} \texttt{(y)})$, is equivalent to the set of \pgkeys{} $\set{\psi_1, \psi_2, \psi_3}$ with:
\begin{align*}
    \psi_1 ={} & \big(S(\texttt{x},\texttt{y'},\texttt{z'}) \Rightarrow~ \MANDATORY(\texttt{x}), S(\texttt{x},\texttt{y},\texttt{z})\land \set{\texttt{x.speaks} = \texttt{y.speaks}}\big) \\
    \psi_2 ={} & \big(S(\texttt{x},\texttt{y},\texttt{z}) \Rightarrow~ \MANDATORY(\texttt{y.speaks})\big) \\
    \psi_3 ={} & \big(S(\texttt{x},\texttt{y'},\texttt{z'}) \Rightarrow~ \SINGLETON(\texttt{y.speaks}), S(\texttt{x},\texttt{y},\texttt{z})\big)
\end{align*}
Intuitively, $\psi_1$ requires that \texttt{x} has a ``\texttt{speaks}'' key and that some member \texttt{y} sharing a forum with \texttt{x} speaks the same language. $\psi_2$ ensures that every such member \texttt{y} has a ``\texttt{speaks}'' key, while $\psi_3$ enforces that all these members speak the same language. Together, this means every \texttt{y} who shares a forum with \texttt{x} speaks the same language as \texttt{x}.

\begin{remark}
    \label{r:npgk-is-ggd-cq-eq}
    An adaptation of \cref{p:mpgk-is-1ggd-cq-eq,p:pgk-sub-ggd-cq-eq} shows that \emph{unconstrained \pgkeys{}} are equivalent to $\ggd{}$.
\end{remark}

\subsection{Separations}
\label{ssec:separations-eq}

The separation results in this section are of two kinds. First, fragments with only one shared variable cannot usually simulate fragments with more shared variables
(\cref{p:ggd-neq-pgk-cq-eq,p:1ggd-neq-gfd-cq-eq}).
Second, fragments with only one shared variable cannot simulate the $\SINGLETON$ and $\EXCLUSIVE$ assertion keywords (Proposition~\ref{p:mpgk-neq-pgk-cq-eq}).
First, we establish $\pgkeys \neq \ggd$.

\begin{propositionrep}
    \label{p:ggd-neq-pgk-cq-eq}
    For~$\crpqeq$, some \ggds{} are not expressible in \pgkeys.
\end{propositionrep}
\begin{proofsketch}
    The proof amounts to showing that $\varphi = (x \rightarrow y \rightarrow z \Rightarrow x \rightarrow z)$ cannot be expressed in \ggd, and uses ideas similar to those used in \cref{t:nggd-sub-mggd}. One simply has to be careful that the additional $\EXCLUSIVE$ and $\SINGLETON$ keywords cannot be used to simulate extra shared variables. 
    A complete proof can be found in \cref{a:p:ggd-neq-pgk-cq-eq}.
\end{proofsketch}
\begin{proof}
    \label{a:p:ggd-neq-pgk-cq-eq}
    Let $\varphi = (x \rightarrow y \rightarrow z \Rightarrow x \rightarrow z)$ be a \ggd{}, and $\Psi$ be a set of $\pgkeys$.
    Let~$N = \max\limits_{\psi \in \Psi} |\var_\psi| + 1$, where $\var_\psi$ is the set of variables of $\psi$. Let $K_1$ and $K_2$ be cliques that share exactly one vertex $u$ with $|K_1| = |K_2| = N$. Let $G_1 = K_1$ and $G_2 = K_1 \cup K_2$ be graphs.
    One can notice that $G_1 \models \varphi$ and $G_2 \not\models \varphi$. Indeed, if $v_1 \neq u$ is a vertex of $K_1$ and $v_2 \neq u$ a vertex of $K_2$, then $v_1 \rightarrow u \rightarrow v_2$, but there is no edge from $v_1$ to $v_2$.

    We now show that $G_1 \models \Psi$ implies that $G_2 \models \Psi$. Let $\psi \in \Psi$. Then $\psi$ is of the form:
    \begin{equation}
        \psi ={} \big(S(x, \bar{y}) \Rightarrow \alpha(\bar{\nu}), T(x, \bar{z})\big)
    \end{equation}
    
    Assume that $G_1\models \psi$. The only interesting case is $\alpha =~ \MANDATORY$.
    Indeed, because $G_1$ is a clique, if $\alpha =~ \EXCLUSIVE$, then the key expression of $\psi$ is trivially satisfied if $x\in\bar\nu$ and trivially unsatisfiable if $x\notin\bar\nu$. Similarly, if $\alpha =~ \SINGLETON$, then the key expression of $\psi$ is trivially satisfied if $\bar\nu = (x)$ and trivially unsatisfiable otherwise. Note that an unsatisfiable key expression does not imply $G_1 \not\models \psi$; it only means $S$ or $T$ has no match, which directly yields $G_2 \models \psi$, completing the proof.

    Assume now that $\alpha =~ \MANDATORY$. Assume that there exists a match $h_s^{(2)}$ of $S$ over $G_2$. Remark that the image of $h_s^{(2)}$ is a graph of size less than $N$. Since $K_1$ is a clique, it follows that the image of $h_s^{(2)}$ is isomorphic to some subgraph of $K_1$. Hence, there exists a match $h_s^{(1)}$ of $S$ over $G_1$. Since $G_1 \models \psi$, this implies that there exists a match $h_t^{(1)}$ of $T$ over $G_1$ which covers $\bar \nu$.

    Let $i\in \set{1,2}$ such that $h_s^{(2)}(x)$ belongs to $K_i$. Again, since the image of $h_t^{(1)}$ is a graph of size less than $N$, it is isomorphic to some subgraph of $K_i$, for which we can arbitrarily choose that $h_t^{(1)}(x)$ is mapped to $h_s^{(2)}(x)$ since all vertices of $K_i$ have symmetric roles. Hence, there exists a match $h_t^{(2)}$ of $T$ over $G_2$ such that $h_t^{(2)}(x) = h_s^{(2)}(x)$. Thus, $G_2 \models \psi$, which concludes the proof.
\end{proof}

\begin{toappendix}
    
\begin{remark}\label{r:connectedness}

The proof of \cref{p:ggd-neq-pgk-cq-eq} may seem unnecessarily complicated. Indeed, consider the following 2\ggd{} that states that a graph is a clique $\varphi' = ((x)(y) \Rightarrow (x) \rightarrow (y))$, where $(x)(y)$ is the disconnected \cq{} with two vertex variables and no edge variable. One can show that $\varphi'$ cannot be expressed in \pgkeys{}, as
no $\pgkey$ can distinguish a large enough clique from two disjoint copies of it.

However, the underlying graph of the source query of $\varphi'$ is disconnected. Connectedness and its effect on expressive power is a natural future work direction. Thus, it is of particular interest to us that the separation of \ggd{} and \pgkeys{} still holds for connected patterns.
\end{remark}

\end{toappendix}

From \cref{p:closure-induced-subgraph:gfd,p:closure-induced-subgraph:exist} of \cref{p:closure-induced-subgraph}, we already know that there are constraints that are expressible in 
$1\ggd$ but not in $\gfd$. Let us show the converse.

\begin{proposition}
    \label{p:1ggd-neq-gfd-cq-eq}
    For~$\crpqeq$, some constraints are expressible in~$\gfd$ but not in~$1\ggd$.
\end{proposition}

\begin{proof}
    \label{a:p:1ggd-neq-gfd-cq-eq}
    We set the \gfd~$\varphi = \big((x)(y), \emptyset \Rightarrow \{x.a = y.a\}\big)$, and aim to show that it is not expressible in $1\ggd$.
    Let $\Psi$ be a set of $1\ggd$.
    For every~$\psi \in \Psi$, let~$\val(\psi)$ denote the set of constants that appear in~$\psi$, and $\val(\Psi) = \bigcup\limits_{\psi \in \Psi} \val(\psi)$.
    Let $G_1 = (\{1\}, \emptyset, \emptyset, \lambda_1, \pi_1)$ and $G_2 = (\{2\}, \emptyset, \emptyset, \lambda_2, \pi_2)$ be graphs where $\lambda_1(1) = \lambda_2(1) = \emptyset$, $\pi:~(1, a) \mapsto \alpha$, $\pi_2:~(2, a) \mapsto \beta$ with $\alpha, \beta \not\in \val(\Psi)$. Finally let $G_3$ be the union of $G_1$ and $G_2$.
    
    One can notice that~$G_1 \models \varphi$, $G_2 \models \varphi$ and~$G_3 \not\models \varphi$. 
    Let us show that if~$G_1 \models \Psi$ and $G_2\models \Psi$, then~$G_3 \models \Psi$, which proves that $\Psi$ cannot be equivalent to $\varphi$.

    Let~$\psi = (S(x, \bar{y}) \Rightarrow T(x, \bar{z})) \in \Psi$ and $h_s^{(3)}$ be a match of $S$ over $G_3$, and assume that $h_s^{(3)}(x) = 1$. The other case is proved in a similar way, by replacing $1$ with $2$, since $\alpha$ and $\beta$ are arbitrary values which do not occur in $\val(\Psi)$.

    Since $G_3$ contains no edge, this means that all variables $x$ and $\bar y$ of $S$ are disconnected vertex variables. We define $h_s^{(1)}$ as the mapping that sends $x$ and $\bar y$ to $1$. Remark that $h_s^{(1)}$ is a match of $S$ over $G_3$. Indeed, since $S$ only contains equality predicates, and those predicates hold in $h_s^{(3)}$, they also hold in $h_s^{(1)}$, which makes more equalities true.

    Now remark that $h_s^{(1)}$ is also a match of $S$ over $G_1$. Thus, there exists a match $h_t^{(1)}$ of $T$ over $G_1$ that is consistent with $h_s^{(1)}$. Since $G_3$ contains $G_1$, $h_t^{(1)}$ is also a match of $T$ over $G_3$ that is consistent with $h_s^{(1)}$. Finally, since the only variable that occurs both in $h_s^{(1)}$ and $h_t^{(1)}$ is $x$ and $h_s^{(1)}(x) = h_s^{(3)}(x)$, then $h_t^{(1)}$ is also consistent with $h_s^{(3)}$. Thus $G_3 \models \psi$.
\end{proof}

The assertion keywords $\EXCLUSIVE$ and $\SINGLETON$ were introduced in \pgkeys{} \cite{DBLP:conf/sigmod/AnglesBDFHHLLLM21}.
They are expected to affect expressive power, and we show that this is indeed the case for \crpqeq. There are \pgkeys{} using $\EXCLUSIVE$ or $\SINGLETON$ that cannot be expressed in \mpgkeys, i.e., with constraints that all use $\MANDATORY$.

\begin{proposition}\label{p:mpgk-neq-pgk-cq-eq}
    For~$\crpqeq$, some \pgkeys{} are not expressible in \mpgkeys.
\end{proposition}

\begin{proof}
    Let~$G_1$ be a graph without data with two vertices~$u, v$ and an edge~$e_1$ from~$u$ to~$v$, and~$G_2$ be a graph without data with three vertices~$u, v, w$, an edge~$e_1$ from~$u$ to~$v$ and an edge~$e_2$ from~$u$ to~$w$. Let $\varphi = (x \Rightarrow~ \SINGLETON(y), x \rightarrow y)$. One can notice that~$G_1 \models \varphi$ and~$G_2 \not\models \varphi$. Let~$\Psi$ be a set of \mpgkeys{}, which we can equivalently assume to be a set of $1\ggds$ by \cref{p:mpgk-is-1ggd-cq-eq}, and let $\psi = \big(S(x, \bar{y}) \Rightarrow T(x, \bar{z})\big)\in \Psi$.
    
    We now show that $G_2 \not\models \psi \implies G_1 \not\models \psi$. If~$G_2 \not\models \psi$, then there exists a match~$h_s^{(2)}$ of $S$ over~$G_2$ for which there is no consistent match~$h_t^{(2)}$ of $T$. Notice that there is no path of length~$2$ or more in~$G_2$ hence path variables of $S$ can be matched only to paths of length 1: either~$e_1$ or~$e_2$. We define~$h_s^{(1)}$ in \eqref{eq:hsht}, on the left, for all~$y \in \Var_S$.
    \begin{align}\label{eq:hsht}
        h_s^{(1)}(y) = \begin{cases}
            v & \text{if } h_s^{(2)}(y) = w \\
            e_1 & \text{if } h_s^{(2)}(y) = e_2 \\
            h_s^{(2)}(y) & \text{else.}
        \end{cases}
        &&
        h_t^{(2)}(y) = \begin{cases}
        u   & \text{if } h_t^{(1)}(y) = u \\
        w   & \text{if } h_t^{(1)}(y) = v \\
        e_2 & \text{if } h_t^{(1)}(y) = e_1
        \end{cases}
    \end{align}
    As both graphs are without data and identifier equalities are preserved, $h_s^{(1)}$ is a match of $S$.
    
    If there is no match of $T$ over~$G_1$ consistent with~$h_s^{(1)}$, then $G_1 \not\models \psi$. Otherwise, let~$h_t^{(1)}$ be such a match. If $h_s^{(2)}(x) \in \set{u, v, e_1}$, then $h_t^{(1)}$ is a match of $T$ over~$G_2$ consistent with~$h_s^{(2)}$. If $h_s^{(2)}(x) \in \set{w, e_2}$, we define~$h_t^{(2)}$ as in \eqref{eq:hsht}, right, for all~$y \in \Var_T$.
    It is a match of $T$ over~$G_2$ because both graphs are without data and equalities between identifiers are still true. Moreover it is consistent with~$h_s^{(2)}$ by definition. In both cases, we have~$G_2 \models \psi$, which is a contradiction. Finally: $G_1 \not\models \psi$.
\end{proof}

Finally, as a curiosity, we show that $\gfd$ for $\crpqeqc$ does not capture all constraints that can be expressed in both 1\ggd{} and \gfd{}. See \cref{a:p:gfdc-neq-gfd-1ggd} for the proof.

\begin{propositionrep}
    \label{p:gfdc-neq-gfd-1ggd}
    There are constraints that are expressible in $1\ggd$ and in \gfd{} for $\crpqeq$, but not in \gfd{} for $\crpqeqc$.
\end{propositionrep}
\begin{proof}
    \label{a:p:gfdc-neq-gfd-1ggd}
    Let $\varphi = \big(x \Rightarrow  \set{x.a = x.b}\big)$ be a $1\ggd$ for \crpqeq. Remark that $\varphi$ is also a $\gfd$ for \crpqeq. Let $\Psi$ be a set of $\gfd$ for \crpqeqc.
    For every~$\psi \in \Psi$, let~$\val(\psi)$ denote the set of constants that appear in~$\psi$,
    and $\val(\Psi)$ be the set of constants that appear in some~$\psi\in\Psi$.
    Let~$G_1 = (V, E, \eta, \lambda, \pi)$ and $G_2 = (V, E, \eta, \lambda, \pi')$ be graphs with
    $V = \set{1}$, $E = \emptyset$,
    $\lambda(1) = \emptyset$, $\pi(1,a) = \pi(1,b) = \pi'(1,a) = \alpha$ and $\pi'(1,b) = \beta$ where $\alpha, \beta \in \val\setminus\val(\psi)$.
    One can notice that $G_1 \models \varphi$ and $G_2 \not\models \varphi$.

    Let $\psi = (S(\bar{x}, \bar{y}) \Rightarrow C_t(\bar{x})) \in \Psi$. We show that $G_1 \models \psi \implies G_2 \models \psi$.
    Let $h_2$ be a match of $S$ over $G_2$. Let $h_1$ be a copy of $h_2$ over $G_1$. This is well defined, because both graphs share the same vertices, edges, and labels. Moreover, $h_1$ satisfies the conditions of $S$, because $\alpha$ and $\beta$ do not appear in $\val(\psi)$, so all $=_c$ predicates return the same truth value over $\alpha$ and $\beta$. Since $G_1 \models \psi$, we know that $h_1 \models C_t$. By the same argument, we deduce that $h_2 \models C_t$. Thus $G_2 \models \psi$ and $G_2 \models \Psi$. Therefore, no set of $\gfd$ for $\crpqeqc$ can accept $G_1$ without accepting $G_2$, hence $\varphi$ is not expressible in $\gfd$ for $\crpqeqc$.
\end{proof}

The results of this section yield the following main theorem, summarized in Figure~\ref{fig:venn-eq}.

\begin{theorem}
\label{thm:sec4}
For $\crpqeq$, $1\ggd{}$ and $\gfd$ are strictly included in $\pgkeys$, which is strictly included in $\ggd$.
\end{theorem}

\section{Queries with Equality and Inequality}\label{sec:cq:neq}
In this section, we consider the case of $\crpqneq$. Results for this case are summarized in~\cref{thm:sec5} and in~\cref{fig:venn-neq}, with all inclusions shown on the figure being strict. While most results readily translate from previous statements, having access to $\neq$ still leads to some unexpected findings. In particular, it makes it possible to simulate the $\SINGLETON$ and $\EXCLUSIVE$ keywords of $\pgkeys{}$ even in fragments that are limited to one shared variable, leading to $\pgkeys = 1\ggd{}$ (\cref{p:mpgkey=pgkey}). In this case, the $\EXCLUSIVE$ and $\SINGLETON$ keywords add no expressive power over $\MANDATORY$, i.e., $\pgkeys = \mpgkeys$.

\begin{propositionrep}
    \label{p:mpgkey=pgkey}
    For~$\crpqneq$,
    $1\ggd$, $\mpgkeys$ and $\pgkeys$ are all equivalent.
\end{propositionrep}
\begin{proofsketch}
    The equivalence between $1\ggd$ and $\mpgkeys$ follows by a proof similar to that of \cref{p:mpgk-is-1ggd-cq-eq} and, by definition, an $\mpgkey$ is a $\pgkey$.
    Hence, it is enough to show that any \pgkey{} using $\EXCLUSIVE$ or $\SINGLETON$ can be simulated by several \pgkeys{} using $\MANDATORY$ only.
    Now, let~$\varphi ={} \big(S(x, \bar{y}) \Rightarrow \alpha(\bar \nu), T(x, \bar{z})\big)$ with~$\alpha \in \set{\EXCLUSIVE, \SINGLETON}$ and $\bar{\nu} = (\nu_1, \ldots, \nu_k)$. 
    We let $\bar{\nu}'$ be a copy of $\bar{\nu}$ where each variable $t$ is replaced by $t'$, and we next set the $\mpgkeys$ $\psi^{(1)}$ and $\psi^{(2)}_1, \cdots, \psi^{(2)}_k$.
    \begin{align*}
        \psi^{(1)} ={} & \big(S(x, \bar{y}) \land S(x', \bar{y}') \land T(x, \bar{z}) \land T(x', \bar{z}') \land \set{x \neq x', \bar{\nu} = \bar{\nu}'} \Rightarrow~ \MANDATORY(x), \bot\big) \\
        \psi^{(2)}_i ={} & \big(S(x, \bar{y}) \land T(x, \bar{z}) \land T(x, \bar{z}') \land \set{\nu_i \neq \nu'_i} \land \ex(\bar{\nu}, \bar{\nu}') \Rightarrow~ \MANDATORY(x), \bot\big)
    \end{align*}
    If $\alpha =~ \EXCLUSIVE$ then $\psi^{(1)}$ is an equivalent \mpgkey{} for $\varphi$, and if $\alpha =~ \SINGLETON$ then $\set{\psi^{(2)}_1, \ldots, \psi^{(2)}_k}$ is an equivalent set of \mpgkeys{} for $\varphi$.
    Indeed, if~$\alpha =~ \EXCLUSIVE$ and~$G$ is a graph such that~$G \not\models \varphi$, then there exist witnesses~$(h_s, h_t)$ and~$(h'_s, h'_t)$ of the violation of~$\varphi$ by~$G$, and one can build the following match~$h$ of the scope of~$\psi^{(1)}$, that makes the descriptor of~$\psi^{(1)}$ unsatisfiable:
    $h(z) = \begin{cases}
        \begin{aligned}
            h_s(z) & \text{ if $z \in \bar{x}$} & \qquad & h'_s(z) & \text{if $z \in \bar{x}'$} \\
            h_t(z) & \text{ if $z \in \bar{y}$} & \qquad & h'_t(z) & \text{if $z \in \bar{y}'$}
        \end{aligned}
    \end{cases}$
    
    For $\alpha =~ \SINGLETON$, the proof is similar; see \cref{a:p:mpgkey=pgkey}.
\end{proofsketch}
\begin{proof}
    \label{a:p:mpgkey=pgkey}
    The equivalence between $1\ggd$ and $\mpgkeys$ follows by a proof similar to that of \cref{p:mpgk-is-1ggd-cq-eq} and, by definition, an $\mpgkey$ is a $\pgkey$.

    We now show how to rewrite a $\pgkey{}$ $\varphi$ into $\mpgkeys{}$. There are two cases, according to the assertion keyword in $\varphi$.

    \noindent {\bf Case 1: $\EXCLUSIVE$.} In that case, $\varphi$ is of the form:
    \begin{align*}
        \varphi ={} \big(S(x, \bar{y}) \Rightarrow~ \EXCLUSIVE(\bar \nu), T(x, \bar{z})\big)
    \end{align*}
    We define $\psi$ as the following $\mpgkey{}$:
    \begin{align*}
        \psi^{(1)} ={} & \big(S(x, \bar{y}) \land S(x', \bar{y}') \land T(x, \bar{z}) \land T(x', \bar{z}') \land \set{x \neq x', \bar{\nu} = \bar{\nu}'} \Rightarrow~ \MANDATORY(x), \bot\big)
    \end{align*}
    where $\bar \nu'$ is a copy of $\bar \nu$ in which each variable $z$ is replaced by $z'$.
    
    Let us show that $\psi$ is equivalent to $\varphi$. Let $G$ be a graph. Assume that $G \not \models \varphi$. Then there exist two witnesses $(h_s,h_t)$ and $(h_s',h_t')$ for $\varphi$ such that $h_s(x) \neq h_s'(x)$ and $h_t(\bar \nu) = h_t'(\bar \nu)$. From these witnesses, we can build a match $h$ of the scope of $\psi$ as follows:
    $$h(z) = 
        \begin{cases}
            h_s(z) &\text{ if } z\in \bar x \\
            h_t(z) &\text{ if } z\in \bar y \\
            h_s'(z) &\text{ if } z\in \bar x'\\
            h_t'(z) &\text{ if } z\in \bar y'
        \end{cases}
    $$
    One can check that $h$ is indeed a match of the scope of $\psi$ by definition of $(h_s,h_t)$ and $(h_s',h_t')$. However, the descriptor of $\psi$ is not satisfiable. Thus $G \not\models \psi$. 
    
    The converse direction follows from similar arguments: a match of the scope of $\psi$ can be split into two witnesses for $\varphi$ that violate the $\EXCLUSIVE$ assertion.

    \noindent {\bf Case 2: $\SINGLETON$.} In that case, $\varphi$ is of the form:
    \begin{align*}
        \varphi ={} \big(S(x, \bar{y}) \Rightarrow~ \SINGLETON(\bar \nu), T(x, \bar{y})\big)
    \end{align*}
    Let $\bar{\nu} = (\nu_1,\ldots,\nu_k)$. We can show with arguments that are similar to Case~1 that $\varphi$ can be rewritten into the set of $\mpgkeys$
    $\Psi = \set{\psi_1, \ldots, \psi_k}$ where:
    \begin{align*}
        \psi_i ={} & \big(S(x, \bar{y}) \land T(x, \bar{z}) \land T(x, \bar{z}') \land \set{\nu_i \neq \nu'_i} \land \ex(\bar{\nu}, \bar{\nu}') \Rightarrow~ \MANDATORY(x), \bot\big)
    \end{align*}
    In other words, $G\not \models \Psi$ if and only if there exist two $\nu_i$ with different values for the same $x$, which precisely captures the $\SINGLETON$ assertion.
\end{proof}

For instance, the \pgkey{} $\varphi$ and \mpgkey{} $\psi$ both enforce that each message has a unique creator.
\begin{align*}
    \varphi = & \big(\texttt{(z:Message)} \Rightarrow~ \SINGLETON(\texttt{x}), \texttt{(x)} \xleftarrow{\texttt{:hasCreator}} \texttt{(z)}\big) \\
    \psi = & \big(\texttt{(z:Message)} ~\land~ \texttt{(z)} \xrightarrow{\texttt{:hasCreator}} \texttt{(x)} ~\land~ \texttt{(z)} \xrightarrow{\texttt{:hasCreator}} \texttt{(x')}, \set{\texttt{x} \neq \texttt{x'}} \\
    & \Rightarrow~ \MANDATORY(\texttt{z}), \set{\texttt{z} \neq \texttt{z}}\big)
\end{align*}

From \cref{t:nggd-sub-mggd}, we know that $1\ggd \subsetneq \ggd$, and we have shown that, for $\crpqneq$, we have $1\ggd = \pgkeys$, hence we have $\pgkeys \subsetneq \ggd$, as stated below.

\begin{corollary}\label{c:ggd-sub-pgk-cq-neq}
    For~$\crpqneq$, we have~$\pgkeys \subsetneq \ggd$.
\end{corollary}

Next, we show that inequality predicates let us simulate any $\gfd$ with a set of $1\ggds$.

\begin{propositionrep}
    \label{p:gfd-sub-1ggq-nq-neq}
    For~$\crpqneq$, we have $\gfd \subseteq 1\ggd$.
\end{propositionrep}
\begin{proofsketch}
    Let~$\varphi = (S(x, y, \bar{z}) \Rightarrow \set{\gamma})$ be a $2\gfd$, and recall that $2\gfd = \gfd$ (\cref{l:split-gfd}). $\gamma$ has one of the following forms: $x.a = y.b$,
    $x.a \neq y.b$, $x = y$, $x \neq y$, $x.a = c$, or~$x.a \neq c$. We consider only the case where~$\gamma$ has the form~$x.a = y.b$. Other cases are similar and can be found in \cref{a:p:gfd-sub-1ggq-nq-neq} as well as a more detailed proof.

    We define the following~$1\ggds$:
    \begin{align*}
        \psi_1 ={}& (S(x, y, \bar{z}) \Rightarrow \set{\ex(x.a)}) &
        \psi_2 ={}& (S(x, y, \bar{z}) \Rightarrow \set{\ex(y.b)}) \\
        \psi_3 ={}& (S(x, y, \bar{z}) \land \set{x.a \neq y.b} \Rightarrow \bot)
    \end{align*}
    Notice that $\psi_1$ and~$\psi_2$ respectively ensure that~$x.a$ and~$y.b$ are defined, while ~$\psi_3$ enforces $x.a = y.b$ when they exists. Hence, $\{\psi_1, \psi_2, \psi_3\}$ is a set of~$1\ggds$ equivalent to~$\varphi$.
\end{proofsketch}
\begin{proof}
    \label{a:p:gfd-sub-1ggq-nq-neq}
    Let $\varphi$ be a $\gfd$. \cref{l:split-gfd} says that $\varphi$ can be assumed to be of the form $\varphi = (S(x, y, \bar{z}) \Rightarrow C_t(x, y))$ with $|C_t| = 1$.
    We show how to construct a set $\Psi$ of $1\ggd$ that is equivalent to $\varphi$.
    Note that $C_t$ has one of the following forms: $\set{x.a = y.b}$,
    $\set{x.a \neq y.b}$, $\set{x = y}$, $\set{x \neq y}$, $\set{x.a = c}$, or $\set{x.a \neq c}$.

    \noindent {\bf Case 1:} $C_t = \set{x.a = y.b}$.
    The set $\Psi$ must ensure that, whenever $S$ is satisfied, then $x.a$ and $y.b$ are both defined and equal. We choose $\Psi = \set{\psi_1, \psi_2, \psi_3}$ where:
    \begin{align}
        \psi_1 ={}& \big(S(x, y, \bar{z}) \Rightarrow (x), \set{\ex(x.a)}\big)\\
        \psi_2 ={}& \big(S(x, y, \bar{z}) \Rightarrow (y), \set{\ex(y.b)}\big) \\
        \psi_3 ={}& \big(S(x, y, \bar{z}) \land \set{x.a \neq y.b} \Rightarrow (x), \bot\big)
    \end{align}

    Let $G = (V, E, \eta, \lambda, \pi)$ be a graph, and $\mathrel{\theta}$ be the binary relation over $V \cup E$ defined by: $e \mathrel{\theta} f$ if there exists a match~$h$ of $S$ such that $h(x) = e$ and $h(y) = f$. We then have the following equivalences:
    \begin{itemize}
        \item $G \models \psi_1$ if and only if for all $e \in \dom(\theta)$, $\pi(e, a)$ is defined;
        
        \item $G \models \psi_2$ if and only if for all $f \in \im(\theta)$, $\pi(f, b)$ is defined;
        
        \item $G \models \psi_3$ if and only if for all $e \in \dom(\theta)$ and $f \in \im(\theta)$,
        either $\pi(e, a)$ or $\pi(f, b)$ is not defined, or $\pi(e, a) = \pi(f, b)$. 
    \end{itemize}
    Hence $G$ satisfies $\Psi$ is and only if for all $e, f$ such that $e \mathrel{\theta} f$:
    $\pi(e, a)$ and $\pi(f, b)$ are defined, and $\pi(e, a) = \pi(f, b)$.
    This is equivalent to $G \models \varphi$. Other cases are similar and listed below. 

    \noindent {\bf Case 2:} $C_t = \set{x.a \neq y.b}$. Then $\Psi = \set{\psi_1, \psi_2, \psi_4}$, with:
    \begin{equation}
        \psi_4 = (S(x, y, \bar{z}) \land \set{x.a = y.b} \Rightarrow (x), \bot)
    \end{equation}
    \noindent {\bf Case 3:} $C_t = \set{x.a = c}$. Then $\Psi = \set{\psi_{=_c}}$, with:
    \begin{equation}
        \psi_{=_c} = (S(x, y, \bar{z}) \Rightarrow (x), \set{x.a = c})
    \end{equation}
    \noindent {\bf Case 4:} $C_t = \set{x.a \neq c}$. Then $\Psi = \set{\psi_{\neq_c}}$, with:
    \begin{equation}
        \psi_{\neq_c} = (S(x, y, \bar{z}) \Rightarrow (x), \set{x.a \neq c})
    \end{equation}
    {\bf Case 5:} $C_t = \set{x = y}$. Then $\Psi = \set{\psi_=}$, with:
    \begin{equation}
        \psi_= = (S(x, y, \bar{z}) \land \set{x \neq y} \Rightarrow (x), \bot)
    \end{equation}
    {\bf Case 6:} $C_t = \set{x \neq y}$. Then $\Psi = \set{\psi_{\neq}}$, with:
    \begin{equation}
        \psi_{\neq} = (S(x, y, \bar{z}) \land \set{x = y} \Rightarrow (x), \bot)
    \end{equation}
\end{proof}

The remaining inclusions and separations follow directly from
\cref{p:gfdc-neq-gfd-1ggd,l:split-gfd,p:closure-induced-subgraph},
yielding the hierarchy below (see also \cref{fig:venn-neq}).

\begin{theorem}
\label{thm:sec5}
For $\crpqneq$, we have that: $\gfd \subsetneq 1\ggd = \pgkeys{} \subsetneq \ggd$.
\end{theorem}

\section{Related Work}
\label{sec:related}
Property graph constraint languages have largely been studied in isolation. \gfd{}~\cite{DBLP:conf/sigmod/FanWX16a} lift relational functional dependencies to account for attributes and graph patterns. \ged{}~\cite{DBLP:journals/tods/FanL19} subsume \gfd{} and keys for graphs~\cite{DBLP:journals/pvldb/FanFTD15} in a uniform formalism by allowing equality over object identifier, a difference we capture through the choice of admissible predicates. \gdd{}~\cite{DBLP:journals/pvldb/KwashieLLLSY19} further add distance and similarity predicates for approximate entity resolution. On the TGD side, \ggd{}~\cite{DBLP:conf/cikm/ShimomuraFY20} allow the head of a constraint to be an arbitrary graph pattern. Other proposals target uniqueness specifically, either for Neo4j's property graph model~\cite{DBLP:conf/er/Link20}, as general integrity constraints in database systems~\cite{DBLP:conf/ant/PokornyVK17}, or as uniqueness constraints over labels and composite properties~\cite{DBLP:journals/jdiq/SkavantzosLZL23,DBLP:conf/caise/SkavantzosZL21}, which underpin normal forms for property graphs~\cite{DBLP:journals/vldb/SkavantzosL25}. While TGD and their fragments have been extensively studied in the relational and RDF settings~\cite{DBLP:conf/lics/CalvaneseGLV00,DBLP:conf/ijcai/GottlobP15}, no comparable analysis exists for property graphs.

\pgkeys{}~\cite{DBLP:conf/sigmod/AnglesBDFHHLLLM21} consolidate this line of work in a single mechanism to identify and reference property graph objects and are, through \pgschema{}~\cite{DBLP:journals/pacmmod/AnglesBD0GHLLMM23}, a candidate for the second edition of GQL. The closest work to ours is that of Manouvrier and Belhajjame~\cite{DBLP:conf/adbis/ManouvrierB24,DBLP:journals/is/ManouvrierB26}, who survey graph dependency formalisms and give mappings translating them into \pgkeys{}.
Hence, their focus is on proposing implementable mappings, whereas ours is on establishing formal expressiveness hierarchies.
Moreover, the target language they use for translations is what we refer to as unconstrained \pgkeys{} (\cref{unconstrained}).
In their setting, \pgkeys{} trivially subsume other formalisms, being equivalent to \ggd{} (\cref{r:npgk-is-ggd-cq-eq}). In fact, our results show this relaxation to be precisely what makes \ggd{} translatable, as \pgkeys{} are strictly weaker than \ggd{} whether or not inequality is available (\cref{p:ggd-neq-pgk-cq-eq}, and \cref{c:ggd-sub-pgk-cq-neq}).

Comparative studies remain scarce. Fan and Lu~\cite{DBLP:conf/sigmod/FanWX16a} compare classes of graph dependencies, but predate both \ggd{} and \pgkeys{}. For schema and shape languages, Ahmetaj et al.~\cite{DBLP:conf/www/AhmetajBHHJGMMM25} provide a common foundation for SHACL~\cite{shacl}, ShEx~\cite{DBLP:conf/i-semantics/PrudhommeauxGS14}, and \pgschema{}, recently extended with property constraints~\cite{DBLP:conf/kcap/TomaszukG25}, and a rich body of work analyses the expressive power of RDF query and navigational languages~\cite{DBLP:conf/www/ArenasCP12,DBLP:journals/tods/PerezAG09}. None of these address dependencies.

\section{Conclusion and Perspectives}
\label{sec:conclusion}
We positioned \pgkeys{} with respect to \gfd{} and \ggd{}, highlighting that the number of shared variables is crucial to expressiveness, while assertion keywords ($\EXCLUSIVE$, $\SINGLETON$) can be simulated when inequality predicates are allowed.
Our results inform future constraint language design and open several research directions detailed below.

\subparagraph{Connectedness}
Query connectedness is motivated by practical considerations. For example, Neo4j issues a warning in the presence of a disconnected pattern, as it results in a Cartesian product and degraded performance~\cite{neo4j-cartprod}. The same observation applies to constraints.
In general, imposing connectedness seems to strictly reduce the expressive power of \ggd{}, \gfd{} and \pgkeys{}, but not necessarily to the same extent.
For instance, we used disconnected constraints to simulate the $\EXCLUSIVE$ keyword in \ggd{} (\cref{p:pgk-sub-ggd-cq-eq,p:mpgkey=pgkey}). Under connectedness, we expect the inclusion not to hold, and~$\pgkeys$ to be incomparable with \ggd{} in expressive power.
Assertion keywords can thus be interpreted as a weak form of disconnectedness: no connected \ggd{} can simulate the \pgkey{} 
$\big(x \Rightarrow~ \EXCLUSIVE(x.a)\big)$, which is connected. The exact impact of connectedness requires further investigation.

\subparagraph{Underlying query language}

Practical languages, such as Cypher~\cite{DBLP:conf/sigmod/FrancisGGLLMPRS18} and GQL~\cite{DBLP:conf/sigmod/DeutschFGHLLLMM22,GQL-ISO}, are far more expressive than the fragment we consider (\crpq{} with (in)equality). For starters they support bidirectional edge navigation (C2\rpq) which, from an expressive power perspective, strictly subsume \rpq, and might collapse some of the constraint fragments.

Moreover, we evaluate \rpq{} and \crpq{} under the so-called \emph{all-walk semantics}, in reference to the $\ALL$ $\WALK$ keywords in GQL~\cite{GQL-ISO}.
Other semantics used in Cypher and GQL (e.g., trail or shortest semantics) could significantly alter the landscape in \cref{fig:inclusion-cqeq-cqneq}.
Indeed, \crpq{} under trail semantics are not stable under homomorphism, 
and \crpq{} are not monotonic under shortest semantics, two key properties used in various proofs.

Cypher and GQL allow to \emph{decompose} path variables.
For instance, the GQL query $Q_s={}${\small\tt{MATCH p = (x)-[e]->(y)-[f]->(z)}} decomposes the path variable~{\small\tt p} into the variables {\small\tt x,e,y,f,z}: the variable~{\small\tt y} (resp.~{\small\tt f}) is necessarily bound to the second vertex (resp.~edge) of the path associated with {\small\tt p}. This effectively removes limits on shared variables.
Indeed, consider the constraint where~$Q_s$ is the source pattern and the target pattern is~$Q_t={}${\small\tt{MATCH p  = (x')-[e']->(y')-[f']->(z')}}.
Although those queries share only one (path) variable~{\small\tt p}, the equalities {\small\tt x}${}={}${\small\tt x'}, {\small\tt e}${}={}${\small\tt e'}, and so on, are implicitly enforced. 
For any connected source pattern, a path can cover all variables and transmit them to the target with a single variable.

\subparagraph{Restrictions on the number of shared variables}

In addition to playing a crucial role for expressive power, the number of shared variables between the left- and right-hand sides of constraints also impacts the complexity of related tasks.
For instance, the validation problem consists in deciding whether a graph $G$ satisfies a set of constraints $\Sigma$. In combined complexity, \gfd{} validation is \conp-complete~\cite{DBLP:conf/sigmod/FanWX16a} and \ggd{} validation is \piip-complete~\cite{DBLP:conf/cikm/ShimomuraFY20}.
To our knowledge, no work parametrizes this complexity by the number of shared variables.
A careful analysis shows that, when~$n$ is fixed, the naive validation algorithm for $n$\ggd{} \cite{DBLP:conf/cikm/ShimomuraFY20} is in \diip, a class believed to be strictly included in \piip. Moreover, the restricted case of 1\ggd{} for \cqeq{} can be expressed in first-order logic with unary negation and bounded negation depth, a fragment whose model-checking can be done in \tiip \cite{DBLP:journals/corr/SegoufinC13}, a subclass of \diip{}. This hints at the fact that the \diip{} bound may not be tight in general. We leave this interesting question for future work.


\bibliography{references.bib}

\end{document}